\documentclass[journal]{IEEEtran}
\usepackage{amsfonts,amssymb}
\usepackage{array}
\usepackage{multirow}
\usepackage[table,x11names]{xcolor}
\usepackage{color}
\usepackage{url}
\usepackage{soul}
\usepackage{epsfig}
\usepackage{graphicx}
\usepackage{float}
\usepackage{mathtools,cuted}
\usepackage{tabularx} 
\usepackage[thmmarks,amsmath]{ntheorem}

\usepackage{lipsum}
\usepackage{setspace}
\usepackage{bm}
\usepackage{hyperref}
\usepackage{authblk}
\usepackage[final]{changes}
\usepackage{wrapfig}
\usepackage[cal=pxtx]{mathalfa}
\usepackage{rotating,threeparttable,booktabs,caption,dcolumn}
\usepackage{ragged2e}
\usepackage{enumitem}
\usepackage{stfloats}
\usepackage[linesnumbered,ruled]{algorithm2e}
\makeatletter
\def\BState{\State\hskip-\ALG@thistlm}
\makeatother
\newcommand{\change}[1]{{\textcolor{black}{#1}}}
\newcommand{\changeR}[1]{{\textcolor{black}{#1}}}

\theoremheaderfont{\normalfont\bfseries}
\theoremseparator{.}

\newcommand*{\QED}{\hfill\ensuremath{\square}}%

\newcolumntype{Y}{>{\centering\arraybackslash}X}

\newcolumntype{b}{Y}
\newcolumntype{s}{>{\hsize=.5\hsize}Y}

\newtheorem{proposition}{Proposition}
\newtheorem{corollary}{Corollary}
\newtheorem{lemma}{Lemma}
\newtheorem{proof}{Proof}

\newtheorem{innercustomppty}{Property}
\newenvironment{customppty}[1]
{\renewcommand\theinnercustomppty{#1}\innercustomppty}
{\endinnercustomppty}

\newtheorem{innerremark}{Remark}
\newenvironment{remark}[1]
{\renewcommand\theinnerremark{#1}\innerremark}
{\endinnercustomppty}

\newtheorem{innerdefinition}{Definition}
\newenvironment{definition}[1]
{\renewcommand\theinnerdefinition{#1}\innerdefinition}
{\endinnerdefinition}

\newcolumntype{C}{>{\centering\arraybackslash}X}

\DeclarePairedDelimiter{\norm}{\lVert}{\rVert}

\author{Ashif~Sikandar~Iquebal,~\IEEEmembership{Student Member,~IEEE,}
        Satish~Bukkapatnam,
        and~Arun~Srinivasa
\thanks{A. S. Iquebal and S. Bukkapatnam are with the Department
of Industrial and Systems Engineering, Texas A\&M University, College Station,
TX 77843 USA  (e-mail: ashif\_22@tamu.edu; satish@tamu.edu).}
\thanks{A. Srinivasa is with Department of Mechanical Engineering, Texas A\&M University, College Station,
TX 77843 USA  (e-mail: asrinivasa@tamu.edu; satish@tamu.edu).}
\thanks{This paper has supplementary downloadable material available at \url{https://ieeexplore.ieee.org} provided by the author. The material includes proof of Propositions 1-4, proof of Corollaries 1-4, and some case studies. This material is 3,182 KB in size.}}

\ifCLASSOPTIONcompsoc
  \usepackage[nocompress]{cite}
\else
  \usepackage{cite}
\fi

\newcommand{\cmt}[1]{}

\makeatletter
\def\footnoterule{\relax%
	\kern-5pt
	\hbox to \columnwidth{\vrule width 0.5\columnwidth height 0.4pt\hfill}
	\kern4.6pt}
\makeatother

\begin{document}
\title{Change detection in complex dynamical systems using intrinsic phase and amplitude synchronization}

\maketitle

\begin{abstract}
We present an approach for the detection of sharp change points (short-lived and persistent) in nonlinear and nonstationary dynamic systems under high levels of noise by tracking the local phase and amplitude synchronization among the components of a univariate time series signal. The signal components are derived via Intrinsic Time scale Decomposition (ITD)--a nonlinear, non-parametric analysis method. We show that the signatures of sharp change points are retained across multiple ITD components with a significantly higher probability as compared to random signal fluctuations. Theoretical results are presented to show that combining the change point information retained across a specific set of ITD components offers the possibility of detecting sharp transitions with high specificity and sensitivity.  Subsequently, we introduce a concept of mutual agreement to identify the set of ITD components that are most likely to capture the information about dynamical changes of interest and define an InSync statistic to capture this local information. Extensive numerical, as well as real-world case studies involving benchmark neurophysiological processes and industrial machine sensor data, suggest that the present method can detect sharp change points, on an average 62\% earlier (in terms of average run length) as compared to other contemporary methods tested.
\end{abstract}


\IEEEpeerreviewmaketitle

\section{Introduction}\label{sec:introduction}
\IEEEPARstart{C}ONVENTIONALLY, \changeR{the detection of anomalies and change points involves testing a hypothesis, $H_o:\bm{\theta=\theta_0}$ against $H_a:\bm{\theta\neq\theta_0}$ over some process parameters $\bm{\theta}$. This implicitly assumes that the process exhibits stationarity. In contrast, real-world processes are inherently nonstationary, i.e., they are continuously changing such as modulations to autocorrelation structure \cite{choi2008sequential} or are piece-wise stationary \cite{wang2018dirichlet}. Examples of nonstationary processes include seismic waves (see Fig.~\ref{fig:example}(a)), speech signals such as the word ``greasy''  as shown in Fig.~\ref{fig:example}(b), etc. While nonstationary processes are time varying, abrupt changes may occur in the structure \cite{choi2008sequential} or the distributional characteristics \cite{wang2018dirichlet} of the underlying processes. For instance, change in the autocorrelation of a slowly varying first order auto-regressive (AR) process as shown in Fig.~\ref{fig:example}(c).}

{\color{black}Critical anomalies may occur when real-world processes transition sharply from one dynamic behavior to the other \cite{basseville1993detection,nikiforov1999quadratic}. These anomalies can either be persistent \cite{wang2014change} (e.g., changes in the covariance structure) or short-lived \cite{barry1993bayesian} {\color{black}that typically do not last beyond the sampling interval of the time series}, as in singularities (or spikes). Timely detection of these sharp transitions to an anomalous behavior is crucial for effective process control. However, the existing change detection models are severely limited in discerning these sharp transitions \cite{nemeth2014sequential}. Much of the existing change detection methods are based on utilizing the amplitude information \cite{gao2000structures}; only a handful of the methods have investigated the instantaneous phase properties of the underlying time series signal \cite{chiang2000phase}.}

\begin{figure*}		
	\centering
 	\includegraphics[width = 1\textwidth]{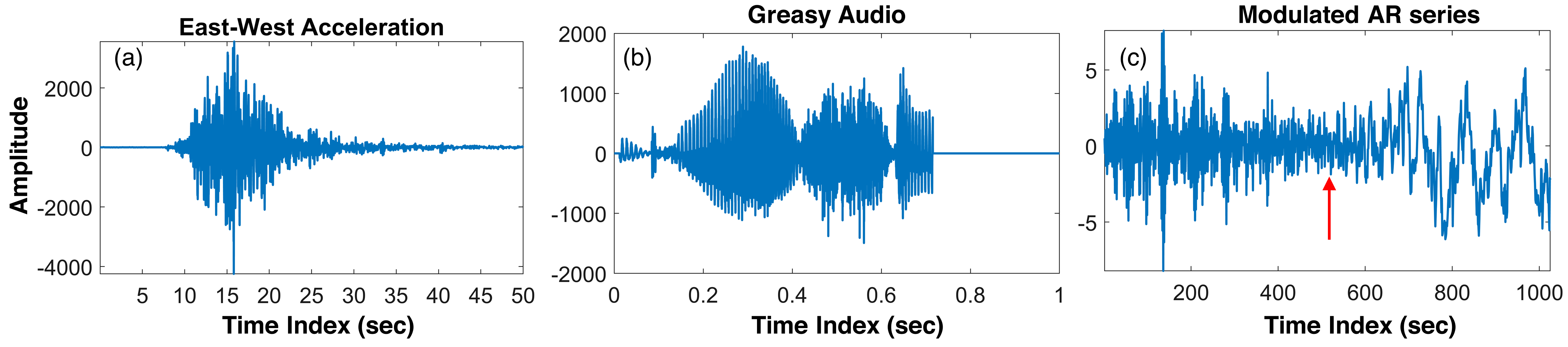}
	\caption{(a) Loma-Prieta Earthquake data (b) ``greasy'' audio signal (c) modulated autoregressive time series with abrupt change at $t=512$ pointed by the arrow.}
	\label{fig:example}
\end{figure*}

The importance of phase is becoming increasingly evident in various domains, such as image analysis and reconstruction \cite{oppenheim1981importance}, electrophysiology \cite{varela2001brainweb}, etc. Phase-based change detection methods have been shown to detect some of the critical events that might go undetected if we rely only on the amplitude information. In many instances, the signal phases exhibited a higher level of synchronization during such events as compared to the amplitudes \cite{rosenblum1996phase}. However, the current phase synchronization approaches require multiple signals (or channels) to utilize the phase information. In the absence of such multi-dimensional time series signals, \change{we propose to} decompose the univariate time series signal $x(t)\in\mathbb{R}, t\in \mathbb{Z}^+$ into multiple components to extract the phase information. 

Unlike stationary Gaussian time series signals, decomposition of nonlinear and nonstationary signals is a non-trivial task \cite{frei2007intrinsic}. Parametric signal decomposition methods such as short-time Fourier, Wavelet or Wigner-Ville transform \cite{Cohen1995} tend to be sub-optimal since they assume an \textit{a priori} basis and often yield poor or inaccurate time-frequency localization. This is a critical issue, especially when detecting short-lived change points. Alternatively, non-parametric methods, e.g., Empirical Mode Decomposition \cite{huang903empirical} or Independent Component Analysis \cite{oja2000independent} offer a data-driven approach with intrinsic basis functions for the decomposition of nonstationary signals. However, these methods cannot be used for real-time applications with streaming data because the decomposition is not causal, i.e., the basis functions are sensitive to the signal length and may change as more data is collected. \changeR{Such decomposition methods are often employed for off-line detection of change points, where the complete time series is available, and the emphasis is on the accuracy of the detection. However, they may not be pertinent for online change point detection, where the objective is to detect the change point as quickly as possible while satisfying the constraints on the false alarms.} Frie and Osorio's Intrinsic Time Scale Decomposition (ITD)  overcomes many of these limitations \cite{frei2007intrinsic}. The elementary decomposition step in ITD considers the signal segment only between consecutive extrema. This allows for a real-time signal decomposition approach with temporal localization of signal features, e.g., change points, across multiple decomposition levels. 

In this paper, we present an approach to detect sharp changes (short-lived and persistent) in noisy real-world processes based on combining the phase and amplitude synchronization among multiple levels of ITD components. The specific contributions of this paper are:  

\changeR{\begin{enumerate}
	\item We derive theoretical results to show that sharp change point features are retained across two or more ITD components with probability more than $0.95$. In contrast, this probability is less than $0.05$ for random signal features. This increases the sensitivity and specificity for detecting change points.  
	\item We introduce a concept of mutual agreement to identify a set of ITD components that are likely to retain the change point information. Subsequently, we develop a change detection statistic called InSync that fuses the phase and amplitude information from ITD components identified via mutual agreement.
	\item We perform extensive numerical and real-world case studies involving short-lived and persistent changes under various process conditions to establish the performance of the present method. We also present the significance level and the rejection criteria (specified in terms of average run length (ARL)) for detecting changes based on the InSync statistic.
\end{enumerate}}

\changeR{We compare the performance of our method with those of conventional approaches, mainly Exponentially Weighted Moving Average (EWMA) and rather contemporary methods including, Wavelet based CUSUM (WCUSUM) method that employs wavelets coefficients to determine the optimal monitoring levels \cite{guo2012multiscale} and Dirichlet Process Gaussian State Machine (DPGSM) \cite{wang2018dirichlet} where a time series is modeled as a mixture of Gaussian and changes occur when the process transitions from one state to the other. We also compare our results with two benchmark change detection packages, \textit{CPM} \cite{ross2015parametric} and \textit{changepoint} \cite{killick2014changepoint}, each of which contains the implementation of several state-of-the-art change detection approaches. We also compare the computational complexity of these algorithms in the context of online change detection.} 

\changeR{The remainder of this paper is organized as follows: in Section 2, we present the ITD algorithm and discuss its relevant properties. In Section 3, we introduce the concepts of intrinsic phase synchronization and mutual agreement followed by the InSync statistic. Section 4 presents the case studies and comparative results followed by concluding remarks and a brief discussion on the performance and limitations of the proposed method in Section 5.}

\section{Overview and properties of ITD}

\hskip 1.5em As noted in the preceding section, we utilize ITD to decompose a signal into different components and use phase and amplitude synchronization among a specific set of ITD components to develop a change detection statistic. In this section, we begin with a brief overview of ITD and analyze the behavior of ITD components at sharp change points.

\subsection{Intrinsic Time Scale Decomposition}
{\color{black}ITD belongs to a general class of Volterra series expansions \cite{Franz:2011} that iteratively extracts the baseline component of a nonlinear or nonstationary signal $x(t)$ such that the residual is a proper rotation, i.e., the successive extrema lies on the opposite side of the zero line \cite{frei2007intrinsic}. Formally, $x(t)$ is decomposed~as:
\begin{equation*}
x(t)  = \mathcal{L}(x(t)) + R(t)
\end{equation*}
Here, $\mathcal{L}(\cdot)$ is the baseline extracting operator such that $\mathcal{L}(x(t)) = L(t)$ is the \textit{baseline component} and $R(t)$ is the residual, referred to as the \textit{rotation component}. 

Let us denote the local extrema of $x(t)$ by $\tau_k, k = 1, 2,\ldots, N$ where $N$ is the total number of local extrema observed in $x(t)$. For simplicity, let $x_k$ and $L_k$ denote $x(\tau_k)$ and $L(\tau_k)$. Then the baseline extracting operator is defined piecewise on the interval $t\in [\tau_k,\tau_{k+1}]$ between successive extrema as:
\begin{equation}\label{eq:0.1}
\mathcal{L}(x(t)) \triangleq L(t) = L_k+\left(\frac{L_{k+1}-L_k}{x_{k+1}-x_k}\right)(x(t)-x_k)
\end{equation}
where 
\begin{equation}\label{eq:0}
L_{k+1} = \frac{1}{2}\left[x_k+\left(\frac{\tau_{k+1}-\tau_k}{\tau_{k+2}-\tau_k}\right) (x_{k+2}-x_k)\right]+\frac{1}{2}x_{k+1}
\end{equation}

Once the input signal $x(t)$ is decomposed into the baseline component and the rotation component, we iterate the decomposition process until a monotonic baseline component is obtained, i.e., 
\begin{eqnarray*}
x(t) & = & \mathcal{L}(x(t)) + R(t) \\
 & = & \mathcal{L}(\mathcal{L}(x(t)) + R(t)) + R(t)\\
 &\ldots&\\ 
 & = & \sum_{j = 1}^{J-1}\mathcal{L}^j (R(t)) + \mathcal{L}^J(x(t))
 \end{eqnarray*}where $\mathcal{L}^J(x(t))$ is the monotonic baseline component obtained after the stopping criteria is reached \cite{frei2007intrinsic} and  $\mathcal{L}^j (R(t))$ is the rotation component at level $j = 1,2,...,J-1$. For simplicity, we denote $\mathcal{L}^j (R(t))$ as $R^j(t)$ and $\mathcal{L}^j (x(t))$ as $L^j(t)$ such that: 
\begin{equation}\label{eq:1}
x(t) = \sum_{j=1}^{J-1} R^j(t) + L^J(t)
\end{equation}
\begin{figure}
	\includegraphics[width=0.48\textwidth]{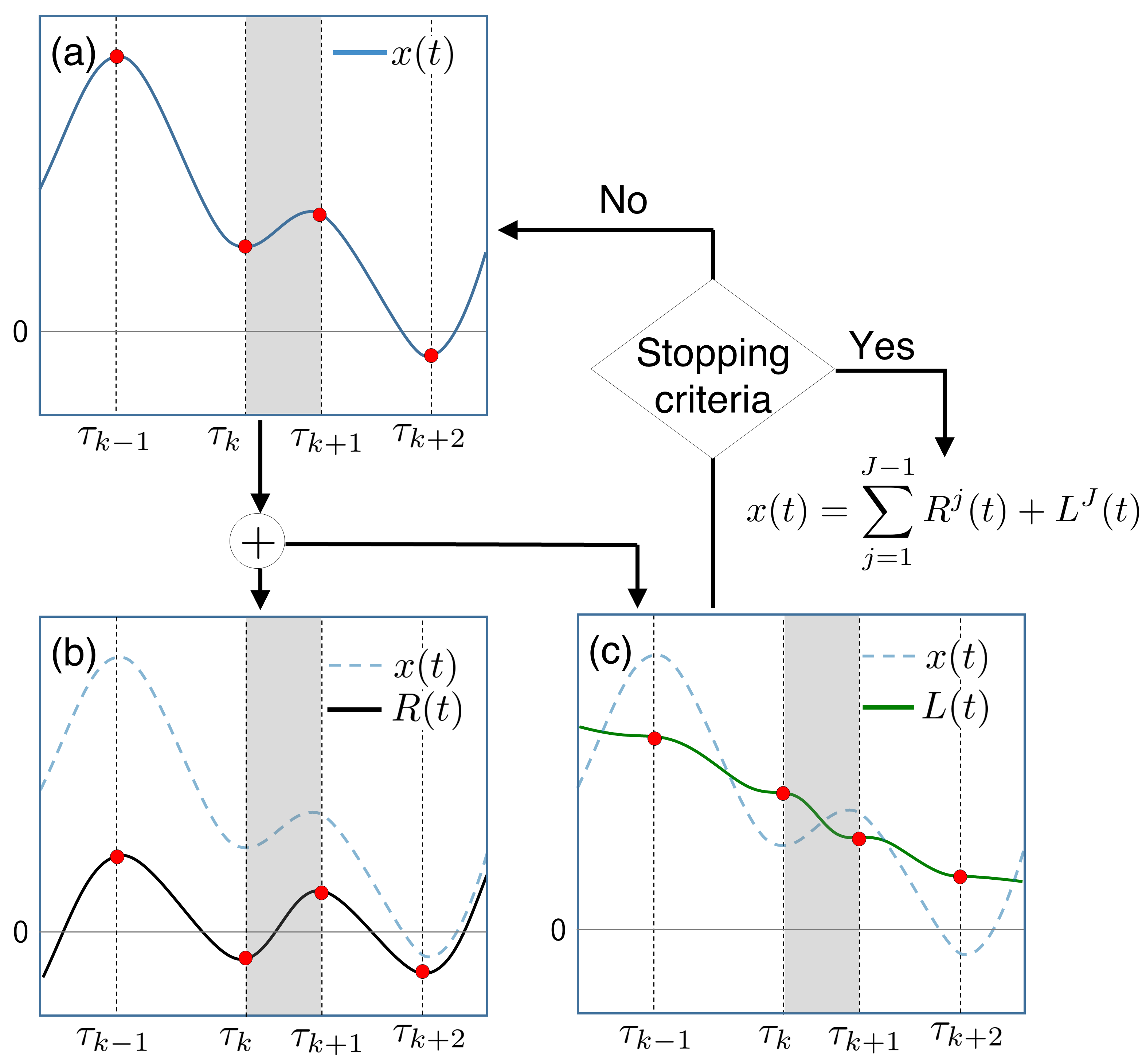}
	\centering
	\caption{Flow chart showing the recursive decomposition of (a) signal $ x(t) $ into (b) a rotation component $R(t)$ and (c) a baseline component $L(t)$. The highlighted region shows the support of the intrinsic basis function defined between two consecutive extrema, $ \tau_{k} $ and $ \tau_{k+1} $.}
	\label{fig:Fl_ch}
\end{figure}

In essence, $R^j(t)$ captures the ``details" of the signal $x(t)$ at the level $j$. The higher the $j$ is, the coarser the details are. As extrema locations are different across different levels, we denote the local extrema at any level $j$ by $\tau^j_k,~k = 1,2,\ldots, N^j$, where $N^j$ is the total number of extrema at level $j$. From an algorithm standpoint, the rotation components are obtained recursively by taking the difference between baseline components obtained at two consecutive levels, i.e., 
\begin{equation}\label{eq:2}
R^j(t)=L^{j-1}(t)-L^j(t),~j=1,2,\ldots,J-1
\end{equation}An instance of the decomposition is shown in Fig.~\ref{fig:Fl_ch}. To initialize the decomposition in the interval $[0,\tau_1]$, we consider the first point of the signal as an extremum (i.e., $\tau_0=0$) and define $L_0 = (x(\tau_0) + x(\tau_1))/2$. }

Since ITD performs the decomposition iteratively between consecutive extrema, {the basis functions have finite support (see Fig.~\ref{fig:Fl_ch}(c))} which allows for (a) {\color{black}causal} representation and (b) temporal localization of the nonlinear and nonstationary signal features across multiple decomposition levels---essential for detecting sharp changes. Also, from Eq.~(\ref{eq:0.1}), we see that the decomposition involves linear operations which can be performed in $\mathcal{O}(cN)$ time, where $N$ is the number of extrema in $x(t)$ and $ c>0$.

\subsection{Properties of ITD}
In this subsection, we extend a simple, yet powerful construct introduced in \cite{restrepo2014defining} to show that the signatures of sharp change points (short-lived as well as persistent) in a given signal are retained across multiple decomposition levels with a specific and significantly higher probability as compared to other random signal features. But first, we present some important properties that will be useful in the representation and understanding of the dynamics of ITD. We begin with a ``half-wave'' representation of the rotation components as introduced in \cite{frei2007intrinsic}. This representation allows for the definition and extraction of instantaneous phase $ \phi_k^j(t) $ and amplitude $ a_k^j$ over finite support. Figure~\ref{fig:im3} shows a representative halfwave defined between the zero crossings $(z_k^{j},z^j_{k+1}]$.

\begin{customppty}{1}\label{one}
	\normalfont Rotation components can be represented as a concatenation of halfwaves $\hslash^j_{k}(t)$ each of which is defined between two consecutive zero crossings $(z_k^{j},z^j_{k+1}]$, for all $k=1,2,\ldots,N^j-1 $, i.e., 
	\begin{equation}
	\hslash^j_{k}(t)\triangleq \lbrace R^j(t)| t\in(z_k^{j},z_{k+1}^{j}]\rbrace
	\end{equation} Here, each of the halfwaves $\hslash^j_{k}(t)$ has a characteristic amplitude $a^j_k$ and an instantaneous phase component $\phi^j_k(t)$. Note that the halfwaves need not be harmonic or even symmetric (i.e., they can be skewed). 
\end{customppty}

\begin{figure}[t]
	\includegraphics[width=0.38\textwidth]{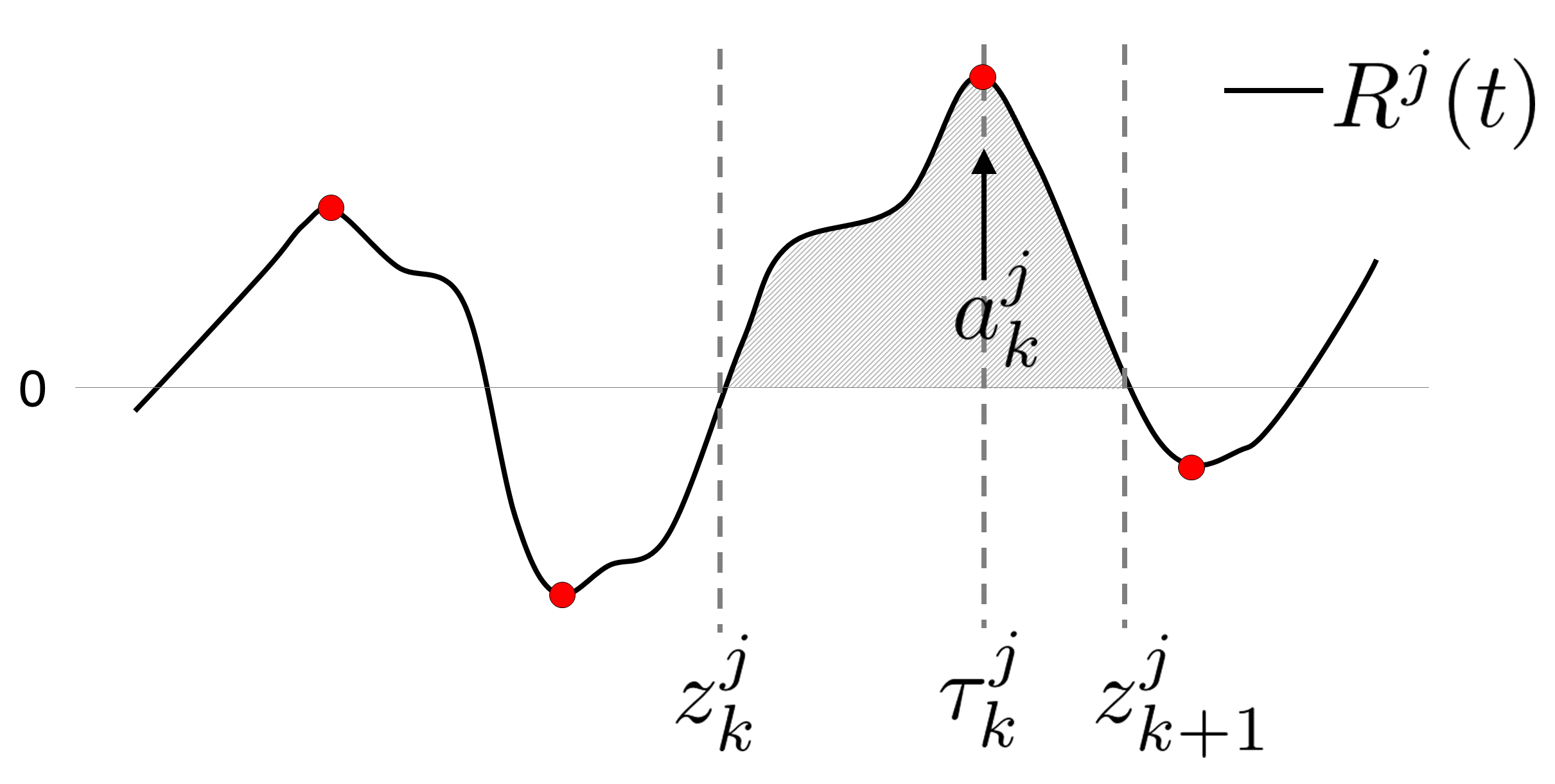}
	\centering
	\caption{A halfwave, $\hslash^j_k(t)$ defined on $(z_k^{j},z_{k+1}^{j}]$ in $R^j(t)$ such that $\tau^j_k$ is the characteristic extrema of $\hslash^j_k(t)$.}
	\label{fig:im3}
\end{figure}


\begin{customppty}{2}\label{two}
	\normalfont \change{The value and the location of extrema in the baseline component at $j+1$, i.e., $\{L(\tau^{j+1}_k),\tau^{j+1}_k\}$ depends only on $\{L(\tau^j_k),\tau^j_k \}$ and not on the signal values elsewhere \cite{restrepo2014defining}.}
\end{customppty}

\noindent Property 2 suggests that in order to obtain $\{L(\tau^{j+1}_k),\tau^{j+1}_k\}$, we do not need to know the entire baseline component $L^j(t)$ but only the value and the location of extrema points in level $j$, i.e., $\{L(\tau^j_k),\tau^j_k \}$. In other words, for the purpose of ITD, we can neglect the intermediate points between two successive extrema. We now employ these properties of ITD to determine the probability with which different change point features in $x(t)$ are retained across the levels of ITD.

{\color{black}We first show that a randomly selected point in the in-control region of $x(t)$ is unlikely to be retained across two or more levels of ITD. Following Property 2, we neglect the intermediate points between the successive extrema of $x(t)$ and consider a time series $x_k$ whose successive samples are the alternating extrema of $x(t)$ (i.e., maximum followed by a minimum). Such an alternating extrema series is a long-term memory process---its autocorrelation function decays slowly regardless of the distribution (or the autocorrelation function) of $x(t)$. Therefore, without loss of generality (also see \cite{restrepo2014defining}), we define: \begin{equation}\label{eq:6}
	x_k \triangleq (-1)^k|w_k|; w_k \sim \mathcal{N}(0,\sigma^2),k\in \mathbb{Z}^+
	\end{equation}  where the magnitudes of successive samples (alternating extrema) are drawn from a white noise process} \change{with mean 0 and standard deviation $\sigma$}. It turns out that the probability that an extremum at $k$ in level $j$ is retained as an extremum across the subsequent $j+\eta, \eta\in\mathbb{Z}^+$ levels decays at a geometric rate as the value of $\eta$ increases. This is presented in the following corollary.

\begin{corollary} \normalfont The probability that an extremum in the rotation component at level $j$ of $x_k$ is retained as an extremum across the subsequent $\eta$ rotation components is approximately equal to $ 0.24^{\eta} $.
\end{corollary}

\noindent The result follows from \cite{restrepo2014defining}. Please refer to Appendix A of the supplementary material for the proof.
\begin{remark}{1}\label{one}
	\normalfont	{\color{black}Corollary 1 shows that the probability with which an extremum in $x_k$ is retained across two or more levels of ITD is less than $0.05$. From a change detection standpoint, it is highly desirable to have a low probability for the extrema points in the in-control region to be retained across multiple levels of ITD. This will enhance the specificity when detecting changes based on combining the information across multiple levels. }
\end{remark}

We now extend this result to a more general scenario by introducing a sharp change in $x_k$ such that the baseline component at some level $j$ is represented as: 
\change{\begin{equation}\label{eq:7}
L^j_k = x_k+\text{sgn}(x_k)\nu\sigma \delta_{k^*}
 \end{equation}
 where $\nu$ is a non-negative scale variable, $\delta_{k^*}$ is Kronecker delta and $\text{sgn}(\cdot)$ is the sign function.} Here, $ \nu\delta_{k^*}$ is representative of a sharp change point at $k^*$. We now determine the probability $P_e(\nu)$ that an extremum at $k^*$ in level $j$ is retained as an extremum in level $j+1$. For notational simplicity, let $r^j_k$ and $l^j_k$ denote $R^j(\tau^j_k)$ and $L^j(\tau^j_k)$. We now write the probability $P_e(\nu)$ as follows: 
\begin{equation}\label{eq:8}
P_e(\nu)=P\left(r^{j+1}_{k^*}-r^{j+1}_{k^*-1}>0\right) P\left(r^{j+1}_{k^*}-r^{j+1}_{k^*+1}>0\right) 
\end{equation}
In the following, we show that as $ \nu $ increases, there is a dramatic increase in the value of $ P_e(\nu) $. For this, we first determine the distribution function of $r_k^{j+1}$.
\begin{proposition}
\normalfont
	\change{Let $ r^{j+1}_k $ be the extremum in the rotation component $R^{j+1}(k)$. The distribution function of the magnitude of $r^{j+1}_{k}$ is given by the convolution of three independent random variables $K_1, K_2$, and $\Gamma$ such that:}
	\begin{equation}\label{eq:9}
	F_{r^{j+1}_{k}}(r)= \hskip -2em \underset{\begin{subarray}{c}\lbrace(\kappa_1,\kappa_2,\gamma) \in\mathbb{R}^2 \times[0,2\nu\sigma];\\ \kappa_1+\kappa_2+\gamma\leq r\rbrace\end{subarray}}{\int\int\int  } \hskip -2em F_{K_1}(d\kappa_1) F_{K_2}(d{\kappa_2})F_{\Gamma}(d\gamma)
	\end{equation}
	where $K_i,i=1,2$ are identically distributed random variables that are a sum of independently distributed normal random variables $l^j_k\sim \mathcal{N}(0,\sigma^2)$ (see Eqs.~\eqref{eq:6} and \eqref{eq:7}) and $\Theta$, i.e., $K_i = l^j_k + \Theta$ with \change{distribution function} given as: 
	\begin{equation*}
	F_K(l,\theta) = \hskip -2em \underset{\{(l,\theta)\in\mathbb{R}^2:l+\theta\leq\kappa\}}{\int\int} \hskip -2em G_{\Theta}(d\theta) G_{l^j_k}(dl)
	\end{equation*} 
	where the distribution function of $\Theta$ is given as:
	$$G_{\Theta}(\theta)=\int_{-\infty}^{\theta}\left(\int_{-\infty}^{\infty}f_{U,l^j_k}\left(l,\frac{\omega}{l}\frac{1}{|l|}dl \right)  d\omega\right)$$
	with $U\sim \text{Uniform}(0,2)$ and $l^j_k\sim \mathcal{N}(0,\sigma^2)$. $\Gamma$ follows a mixture distribution such that:
	$$F_{\Gamma}(\gamma)=\int_{0}^{\gamma}\frac{1}{2\nu\sigma}d\omega \bm{1}_{k={k^*}\pm1}+c \bm{1}_{k=k^*}$$
	where $c>0$ \change{and $\bm{1}_A$ is the indicator function for some set $A$}.
\end{proposition}
Please see Appendix B of the supplementary material for the proof of Proposition 1. In the absence of a closed form representation of Eq.~(\ref{eq:9}), we present the following approximation result to simplify the subsequent analysis:

\begin{corollary}
\normalfont
	 Using Gaussian approximation to the distribution function of $r^{j+1}_{k}$, $P_e(\nu)$ can be deduced in closed form as: 
	\begin{equation}\label{eq:10}
	\hat{P}_e(\nu)=\left[ 1-P\left(\mathcal{Z}\leq-\frac{\nu}{\sqrt{2}} \right)\right]^2  
	\end{equation}
	where $\mathcal{Z}\sim \mathcal{N}(0,1)$.
\end{corollary}

\noindent Proof of the corollary is presented in Appendix C of the supplementary material.

\change{We now use Eq.~(\ref{eq:10}) to determine the probability with which the extremum at $k^*$ in level $j$ is retained as an extremum in level $j+1$. The probability as a function of $\nu$ is shown in Fig.~\ref{fig:im4}(a) in the black line, labeled as ``Approximation''. First, we note that for $\nu=0$, $P_e(\nu)\approx0.25$ is simply is the probability that an extremum in level $j$ is retained as an extremum in level $j+1$ for a white noise signal. This is consistent with the result stated in Corollary 1. Additionally, we note that as $\nu$ increases, there is a sharp increase in the value of $P_e(\nu)$, indicating that the information pertaining to a change point is retained across multiple decomposition levels.} 

\change{To validate the values of $\hat{P_e}(\nu)$ obtained by using the Gaussian approximation, we compare with the corresponding probabilities computed numerically by using the analytical form of the distribution function of $r^{j+1}_k$ as given in Eq.~(\ref{eq:9}) as well as the empirical estimate of $P_e(\nu)$ obtained by using Monte Carlo (MC) simulation. In the MC simulation we perform ITD of $L^j_k$ with different realizations of $x_k$ as given in Eq.~(\ref{eq:7}) and observe the cases when the extremum at $k^*$ is retained as an extremum in the subsequent level. To get a consistent estimate of the probability, we performed 100 MC simulations. From Fig.~\ref{fig:im4}(a) we notice that the Gaussian approximation $\hat{P_e}(\nu)$ closely follows the trend of the probability estimated analytically (blue) as well as via MC simulation (red).}

The sharp rise in $P_e(\nu)$ as observed in Fig.~\ref{fig:im4}(a) can be explained by the Gaussian approximation of $r^{j+1}_k$. First, for $r^{j}_{k^*}$ to remain an extremum in level $j+1$, we need $r^{j+1}_{k^*}-r^{j+1}_{k^*\pm 1}>0$. From the proof of Corollary 2 (Appendix C of the supplementary material), we notice that: $$\left(\hat{r}^{j+1}_{k^*}-\hat{r}^{j+1}_{k^*\pm 1}\right)\sim \mathcal{N}\left(\frac{\nu}{4},\frac{1}{8} \right) $$ Since the mean of $\hat{r}^{j+1}_{k^*}-\hat{r}^{j+1}_{k^*\pm 1}$ is a function of $\nu$, the distribution function of $r^{j+1}_{k^*}-r^{j+1}_{k^*\pm 1}>0$ shifts on the positive $\nu$ axis as the value of $\nu$ increases. As a result, $P_e(\nu)$ increases steeply. 

\begin{remark}{2}\label{two}
	\normalfont	{\color{black}From Fig.~\ref{fig:im4}(a), it is evident that for $\nu \geq3$, the probability of retaining the signatures of change points across two subsequent rotation components is greater than 0.95. This property of ITD to selectively retain the change point information across multiple levels of ITD enhances the specificity and sensitivity of detecting changes based on combining the information across these levels.}
\end{remark}

In the following subsections, we extend the result in Proposition 1 to determine the probability $P_e(\nu)$ in the case of a singularity (short-lived change point) and variance shift (persistent change point). 
\begin{figure}
	\includegraphics[width=0.48\textwidth]{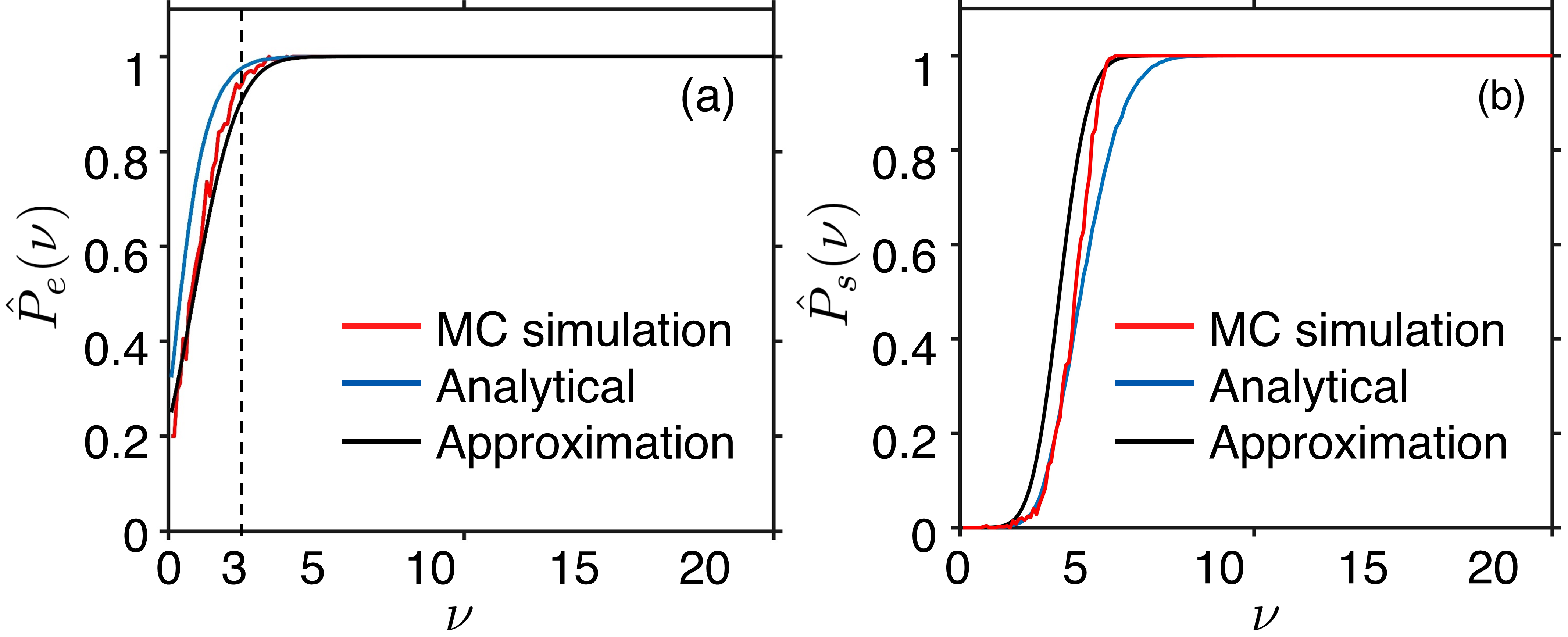}
	\centering
	\caption{Comparison of (a) the probabilities that an extremum at $k^*$ in level $j$ is retained as an extremum in level $j+1$ and (b) the conditional probabilities that a singularity in level $j$ is retained as a singularity in level $j+1$ as obtained via Gaussian approximation (black) with that of the analytical probability estimates (blue) and MC simulation (red).}
	\label{fig:im4}
\end{figure}  

\subsubsection{Extension to singularity detection}
We first consider the systemic feature introduced in Eq.~(\ref{eq:7}) at $k=k^*$ as a singularity whenever $r^j_{k^*}\geq 3\sigma^j$, \change{where $\sigma^j$ is the standard deviation of the rotation component of the signal $y_k$ (as defined in Eq.~(\ref{eq:7})) in level $j$}. The conditional probability $P_s(\nu)$, that a singularity in level $j$ remains as a singularity in level $j+1$, i.e.,
\begin{equation}\label{eq:11}
P_s(\nu) = P\left({r}^{j+1}_{k^*}>3\sigma^{j+1}\big|\nu\geq 3\sigma^j\right)
\end{equation}
can be approximated as given in the following corollary.
\begin{corollary}
\normalfont
	Using the Gaussian approximation to the distribution function of $r^{j+1}_k$ (see Appendix C), $P_s(\nu)$ can be approximated as:
	\begin{equation}\label{eq:12}
	\hat{P_s}(\nu)=1-P\left( \mathcal{Z} \leq 3-\nu \sqrt{19/16}\right)
	\end{equation}
\end{corollary} 

Proof of the corollary is given in Appendix D of the supplementary material. \change{The conditional probability of retaining a singularity as estimated using Eq.~\eqref{eq:12} is shown in Fig.~\ref{fig:im4}(b). We note that for values of $\nu<3$, the probability remains close to 0. However, there is a dramatic increase in $P_s(\nu)$ afterwards and is closes to 1 as the value of $\nu$ exceeds 5. We also compare the probability values obtained from the Gaussian approximation $\hat{P_s}(\nu)$ with that of the analytical estimates as well as MC simulation with 100 runs. The comparative results are presented in Fig.~\ref{fig:im4}(b). We notice that the Gaussian approximation $\hat{P_s}(\nu)$ is consistent with both the probability curves obtained analytically as well as via MC simulation.}

Again, the steep increase in $\hat{P}_s(\nu)$ in Fig.~\ref{fig:im4}(b) can be understood from the Gaussian approximation of ${r}^{j+1}_{k^*}$. Since ${r}^{j+1}_{k^*}\sim \mathcal{N}\left( {\nu}/{4},{1}/{19}\right)$, $P_s(\nu)$ remains close to 0 for ${\nu}/{4}+3\sqrt{{1}/{19}}\leq 3\sigma^j$ and then increases steeply when the above condition no longer holds. This is because the mean of ${r}^{j+1}_{k^*}$ increases linearly as a function of $\nu$ with variance $\ll$ 1. More interestingly, $\nu=3\sigma$ acts somewhat as an ``activation barrier" such that $\hat{P}_s(\nu)\rightarrow1$ as $\nu>3\sigma$. It should also be noted that the conditional probability statement in Eq.~(\ref{eq:11}) is an underestimation of the actual probability with which singularities are retained. This is because there is a small, but non-zero probability that ${r}^{j+1}_{k^*}>3\sigma^{j+1}$ given $\nu<3\sigma^j$, i.e., $P\left({r}^{j+1}_{k^*}>3\sigma^{j+1}\big|\nu < 3\sigma^j\right)>0$ (see Property S1, Appendix F). As a result, a singularity in level $j$ may be retained with much higher probability in the subsequent rotation components than that reported in Eq.~(\ref{eq:11}).

\subsubsection{Extension to variance shift detection}
{\color{black}We now generalize the systemic feature introduced in Eq.~(\ref{eq:7}) to a sharp change in the second moment. Here, we redefine $\bm{l}_k$ such that the second moment of the signal changes sharply at $k=k^*$ from $\sigma_0$ to $\sigma_a$ as: 
\begin{equation}\label{eq:13}
\bm{l}_k\triangleq
\begin{cases}
x_k = (-1)^k|w_k|; w_k \sim \mathcal{N}(0,\sigma_0^2)~,k\leq k^*\\
x_k = (-1)^k|w_k|; w_k \sim \mathcal{N}(0,\sigma_a^2)~,k> k^*\\
\end{cases}
\end{equation}
with  ${\sigma_0^2}/{\sigma_a^2} < 1$. We show in the following that the information of a second-order moment shift (i.e., Eq.~(\ref{eq:13})) is asymptotically retained across the subsequent rotation components. 

\begin{corollary} 
\normalfont The probability $P(\sigma_0, \sigma_a, n_0, n_a)$ that the rotation component $\bm{r}^j$ contains the moment shift is given as:
	$$P({\sigma_0},{\sigma_a}, n_0, n_a) = P\left(\frac{\textit{S}(r^j_{k\leq k^*})}{\textit{S}(r^j_{k> k^*})} < 1\right)= \mathcal{B}_{\varkappa}\left(\frac{n_0}{2},\frac{n_a}{2}\right)$$ where $S(\cdot)$ denotes the sample variance, $\mathcal{B}$ is the regularized incomplete beta function evaluated at $\varkappa = {n_0\sigma_a^2}/({n_0\sigma_a^2 + n_a\sigma_0^2})$ with $n_0+1$ and $n_a+1$ being the length of time series in the in-control and out of control region, respectively.  
\end{corollary}


Refer to Appendix E of the supplementary material for the proof of Corollary 4. The resulting probability map for $P(\sigma_0, \sigma_a, n_0, n_a)$ is shown in Fig.~\ref{fig:varShift}(a) with $n_0 = 1001$. We notice that as $\sigma_a$ increases, the probability of retaining the variance shift information, i.e, $P(\sigma_0,\sigma_a, n_0,n_a)$ in level $j$ asymptotically approaches to 1. In addition to this, Fig.~\ref{fig:varShift}(b) shows the statistical power \change{(i.e., $1-\beta$ where $\beta$ is the probability of type II error)} that rotation component at level $j$ retains the variance shift information at a significance level of $0.05$ also approaches to 1 as the out of control sample size increases.

\begin{remark}{3}\label{three}
	\normalfont	The foregoing results establish the probabilities with which the change point signatures, both for short-lived (Section 2.2.1) and persistent (Section 2.2.2) change points may be retained across the subsequent levels of ITD. The numerical simulations presented in Sections 4.3 and 4.4 investigate the retention of additional forms of short-lived and persistent changes across multiple ITD levels for signals beyond white noise. These studies show the generalizability of the results in real-world applications involving nonlinear and nonstationary signals. 
\end{remark}}

\begin{figure}[t]
	\includegraphics[width=0.48\textwidth]{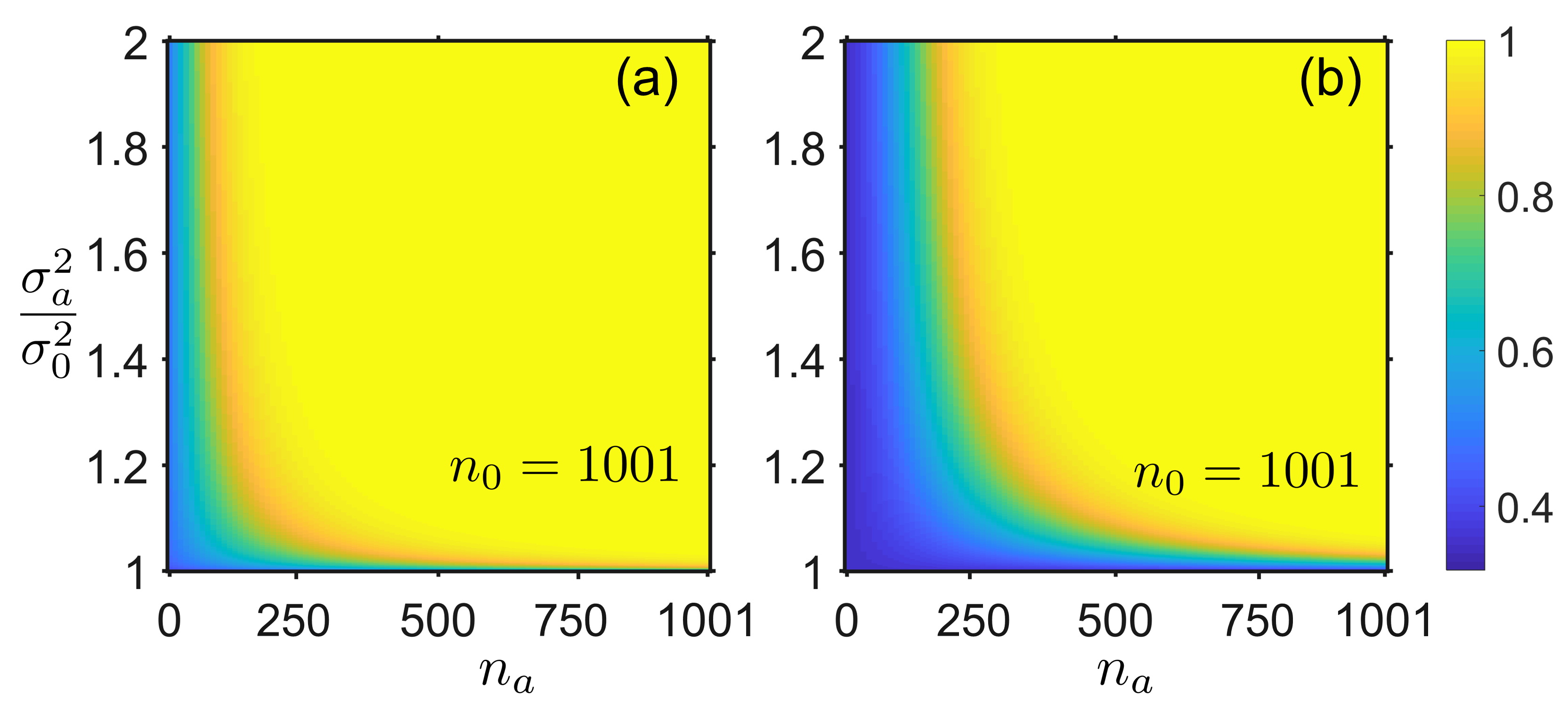}
	\centering
	\caption{ (a) $P(\sigma_0,\sigma_a,n_0,n_a)$ as a function of the ratio of variances, i.e., $\sigma_a^2/\sigma_0^2$ versus the out of control sample size $(n_a)$ with $n_0 = 1001$, (b) variation in the statistical power that the variance shift information is retained in the subsequent rotation component at a significance level of $0.05$.}
	\label{fig:varShift}
\end{figure}

\section{Intrinsic Phase Synchronization}
Evident from the foregoing is that the sharp change point features are highly likely to be retained across multiple levels of $R^j(t)$ as compared to random signal patterns and the specificity of detecting these change points can be significantly enhanced (i.e., false positives can be reduced) if the information from multiple levels are combined. However, not all the components will retain the change point features. The issue of identifying the set $\mathcal{G}$ of rotation components that will retain the change point information still remains. In the following, we employ phase synchronization concepts to determine the set $\mathcal{G}$ of rotation components and subsequently combine the phase and amplitude information that is contained in $R^j(t) \in \mathcal{G}$ to develop the InSync statistic.

\subsection{Phase Synchronization among ITD Components}
\begin{definition}{1}\label{one}
	\normalfont {	Phase synchronization between a halfwave $\hslash^{j_1}_{k}(t)$ of $ R^{j_{1}}(t) $ and the fraction of corresponding halfwave $\hslash^{j_2}_{k}(t)$ at level $j_2>j_1$, within the support, $\text{supp}\left( \hslash^{j_1}_{k}(t)\right) =(z_k^{j_1},z_{k+1}^{j_1}]$ is defined as:
		\begin{equation}\label{eq:14}
		\Phi^{j_1,j_2}_{k}\triangleq\cos(\phi^j_k(t)-\phi^{j+1}_k(t)) = \frac{\left\langle\hslash^{j_1}_{k}(t), \hslash^{j_2}_{k}(t)\right\rangle }{\norm[\big]{\hslash^{j_1}_{k}(t)} \norm[\big]{\hslash^{j_2}_{k}(t)}}
		\end{equation}}
\end{definition}

\noindent The aforementioned definition of phase synchronization is an improvement over the classical phase synchronization quantifier (i.e., $|\phi^{j_1}_k(t)-\phi^{j_2}_k(t)|\approx {\text{constant}}$, \cite{rosenblum1996phase}), in that, it is more robust (due to cosine-scaling) to perturbations in the phases resulting due to noise effects, and provides a more direct quantification of the strength of synchronization between halfwaves at different levels. Comparatively, the classical approach only provides an indirect quantification with an expected value of $|\phi^{j_1}_k(t)-\phi^{j_2}_k(t)|\rightarrow 0$ for highly synchronized halfwaves. Based on this definition, we now estimate the increase in the expected level of phase synchronization when there is a singularity (i.e., $\mathbb{E}[\Phi^{j,j+1}_{k} \big|r^{j}_{k^*}\geq 3\sigma^j] $) versus no singularity (i.e., $\mathbb{E}[\Phi^{j,j+1}_{k} \big|r^{j}_{k^*}< 3\sigma^j] $), as captured in the following proposition:
%

\begin{proposition}
	\normalfont The ratio of expected value of phase synchronization when there is a singularity at $k=k^*$ to the case when there is no singularity at $k=k^*$, i.e.,  
		\begin{eqnarray}\label{eq:15}
		\xi_{k^*}=\frac{\mathbb{E}\left[\Phi^{j,j+1}_{k^*} \big|r^{j}_{k^*}\geq 3\sigma^j\right] }{\mathbb{E}\left[ \Phi^{j,j+1}_{k^*}\big|r^j_{k^*}<3\sigma^j\right] }
		\end{eqnarray}
	is lower bounded as: $$\xi_{k^*} \geq P_s(\nu|\nu>3\sigma^j)\lim_{h\to 0}({P_e(h)})^{-1}\approx4P_s(\nu|\nu>3\sigma^j)$$
\end{proposition}

Please see Appendix F of the supplementary material for the proof. Here, we note that as $P_s(\nu)\rightarrow 1$, we have $\xi_{k^*}\geq 4$. This implies that whenever there is a singularity in $R^j(t)$, the expected level of phase synchronization between the halfwaves at level $j$ and $j+1$ is amplified by more than 4 folds as compared to when there is no singularity. It also suggests that information about a singularity is reflected in the phase synchronization statistic among the corresponding halfwaves.

Apart from singularities, our experimental observations (see Sections 4.3 and 4.4) suggest that the statistic can be used to detect a much broader set of sharp changes in real world dynamic systems. Consistent with the results reported in \cite{gonzalez2002amplitude}, whenever change points are characterized by second/higher order moment shifts in nonlinear and nonstationary dynamic systems, we noted a high level of synchronization between the envelopes of maxima and minima at two different levels of ITD components, also referred to as the Amplitude Envelope Synchronization (AES) (see \cite{banerjee2010chaos} for additional discussion). The value of AES among the levels of rotation components in $\mathcal{G}$ would be higher as compared to those for the remaining components.

\subsection{Mutual Agreement}
As noted in the previous section, it is important to select a set $\mathcal{G}$ of rotation components that would be dynamically similar so that the information contained therein, when fused together, would positively reinforce the presence of a feature or a change point. That is, $\Phi^{j,j+1}_k\rightarrow1$ for $R^j(t),R^{j+1}(t)\in \mathcal{G}$ whenever $k$ is a sharp change point or a feature of interest, and $\Phi^{j,j+1}_k\rightarrow0$ otherwise. This will result in enhanced sensitivity and specificity of detecting change points. We refer to such dynamically similar set of components $\mathcal{G}$ as the set with \textbf{\textit{maximum mutual agreement}}.

{\color{black}In order to determine $\mathcal{G}$ we employ an undirected graph representation $G$ of the rotation components of $x(t)$ such that $G \triangleq \left(V,E\right)$ where the nodes $V=\left\lbrace x(t), R^j{(t)}\right\rbrace$, $j = 0,1,2,\ldots, J-1$ (index $j=0$ represents $x(t)$) and the edges $E = [e_{ij}]=|\Phi^{i,j}|, i\neq j$ capture the pairwise phase synchronization measure (Eq.~(\ref{eq:14})) between the elements of $G$. Here, we consider that the edge weights smaller than a specified Pareto threshold ($\vartheta_p$) represent spurious connections between the elements of $G$ and can be discarded. By adapting the approach presented in \cite{broadwater2010adaptive}, we estimate the threshold $\vartheta_p$ from the realizations of $E$ at a significance level of 10\%, such that $P(E>\vartheta_p)=0.1$. In other words, we consider that only the tail realizations of $E$ capture the salient association between the elements of $G$. As a result of Pareto thresholding, small clusters $\mathcal{G}_1,\mathcal{G}_2,\ldots$ of rotation components are obtained. \changeR{The pseudo code for the selection of rotation components is presented in Algorithm 1. The key step is to remove the edges with edge weights smaller than $\vartheta_p$. In this reduced graph, we identify the connected components each of which forms a cluster. This may be performed by any generic graph search algorithm such as depth first search (Step~6). }

\begin{algorithm}[!h]
	\SetKwInOut{Input}{Input}
	\SetKwInOut{Output}{Output}
	\underline{function mutualAgreement}$(G=(V,E))$\;
	\Input{$V=\{x(t), R^j(t)\},j = 1,\ldots,J-1;E=[e_{ij}]$}   		
	\Output{Clusters $\mathcal{G}_1, \mathcal{G}_2,\ldots$}
	Estimate $\vartheta_p$ such that $P(E>\vartheta_p)\approx 0.1$\;
	Update $E$ as $E[E<\vartheta_p ]\leftarrow 0$\;
	$k\leftarrow 1$\;
	\For{every node $R^j(t)\in V \backslash x(t)$}{ 
		$\mathcal{G}_k\leftarrow$GraphSearch($R^j(t)$) {\color{darkgray}\% \textit{use graph search to identify the components connected to $R^j(t)$}}\;
		$V\leftarrow \{V\backslash \mathcal{G}_k,x(t)$\}\;
		$k\leftarrow k+1$\;
	}
	\caption{Mutual agreement}
\end{algorithm}

For each of the resulting clusters, mutual agreement is defined as: 
\begin{equation}\label{eq:16}
\mathcal{M}^k \triangleq \frac{2}{m(m-1)}\sum_{i,j\in\mathcal{G}_k}^{}|\Phi^{i,j}|
\end{equation}}where $m$ is the cardinality of $\mathcal{G}_k$. An illustrative example of the method is shown in Fig.~\ref{fig:network}. Here, the arc thickness in Fig.~\ref{fig:network}(a) is  proportional to $|\Phi^{i,j}|;\forall i,j=0,1,2,\ldots,J-1,i\neq j$. After thresholding on the realizations of $E$, we obtain three different clusters of rotation components as shown in Fig.~\ref{fig:network}(b). We deem that the cluster that has the maximum value of $\mathcal{M}^k$ (in this case, $\mathcal{G}_1$), contains the set of rotation components with maximum mutual agreement, i.e., $\mathcal{G}\equiv \mathcal{G}_1\sim x(t)=\{R^2(t), R^3(t)\}$.

\begin{figure}[!t]
	\includegraphics[width = 0.5\textwidth]{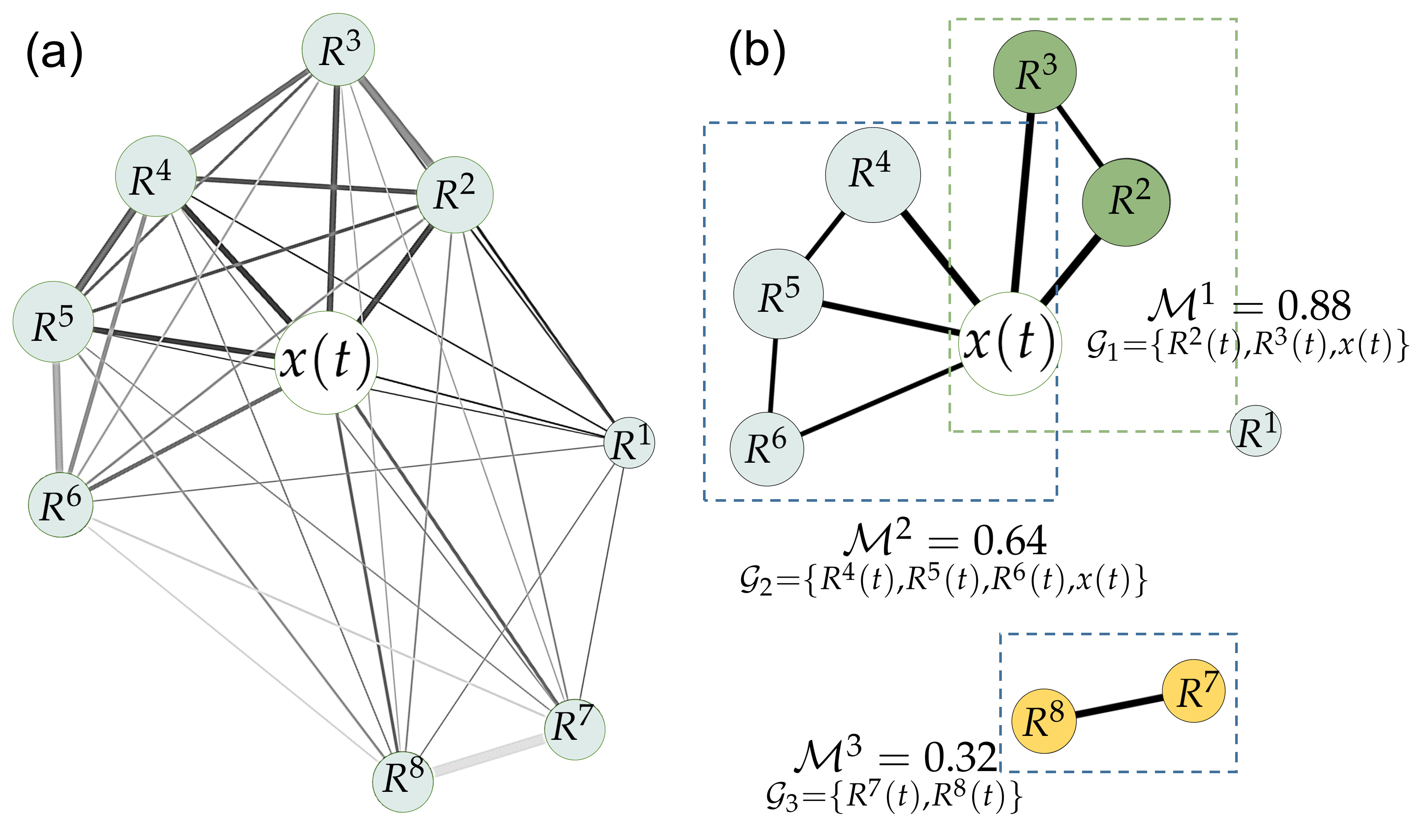}
	\centering
	\caption{(a) Graph representation showing the association between the elements of $G$. (b) Clusters of rotation components obtained after removing the spurious connections as determined by the Pareto threshold $\vartheta_p$.}
	\label{fig:network}
\end{figure} 

\begin{remark}{4}\label{four}
	\normalfont	{\color{black}For a sufficiently long time series, sharp change points such as singularities (and other short-lived anomalies) are mostly captured by lower level rotation components (typically $j\leq 3$) while the change points that are persistent such as trend and moment shifts or other dynamic pattern changes are generally captured by higher level rotation components (typically $j>3$). Whenever multiple change points are present in the signal, e.g., singularities as well as moment shifts, different clusters of rotation components with comparable $\mathcal{M}^k$ values are obtained, each capturing the information pertaining to a specific type of change point. Please see Appendix G of the supplementary material for representative examples.}
\end{remark}

\subsection{The InSync statistic}

In this section, we develop a statistic that would capture and fuse the local phase and amplitude information contained across the rotation components that belong to the set $\mathcal{G}$. Before that, we invoke another property of rotation components that would allow us to combine the phase and amplitude information of the rotation components of $\mathcal{G}$.

\begin{proposition} [\textbf{Property 3}] The support of $\hslash^j_{k}(t)$ at any level $j\geq2$ spans at least one halfwave $\hslash^i_{k}(t)$ from its sub-level $\{R^i(t)\}_{i<j}$ and at most one $\hslash^f_{k}(t) $ from its super-level $\{R^f(t)\}_{f>j}$ as shown in Fig. ~\ref{fig:im6}.
\end{proposition} 

\noindent Proof of this proposition is presented in Appendix H of the supplementary material. This property of ITD components allows us to specify a statistic, referred to as InSync, that fuses the phase and amplitude information contained in the halfwaves across multiple levels within the support of a base halfwave. It may be noted that the InSync statistic inherently borrows the intuition from Kolmogorov's energy cascading principle to combine the phase and amplitude information across multiple levels \cite{paret1997experimental}.
First, we select a base (or reference) rotation component $R^b(t)$ that has the maximum value of weighted degree centrality within the cluster $\mathcal{G}$, i.e.,   
\begin{equation}\label{eq:17}
R^b(t)=\underset{R^j(t)\in\mathcal{G}}{\mathrm{argmax}}
\bigg({\sum_{R^i(t)\in \mathcal{G},i\neq j}^{}|\Phi^{i,j}|}\bigg)
\end{equation}
With this base component $R^b(t)$ determined, we define the InSync statistic as:
\begin{equation}\label{eq:18}
\mathcal{I}(\hslash^b_{k}(t)) \triangleq \bigg( \sum_{R^j(t)\in \mathcal{G}} g\left[\mathcal{E}(\hslash^j_{k}(t))\right] \bigg) \times \prod_{R^j(t)\in \mathcal{G}}\Phi^{b,j}_k
\end{equation}

\noindent such that for each halfwave(s) $\hslash^j_k(t)$, $t\in\text{supp}(\hslash^b_k(t))$. {\color{black}Here, $g(x) = e^{\alpha x}$ is a contrast enhancement function with scale factor $\alpha = {\log(\max(x))}/{\max(x)}$ and $\mathcal{E}(\hslash^b_k(t))$ is the energy (sum of the squares) of each $\hslash^j_{k}(t), t\in\text{supp}(\hslash^b_k(t)), \forall R^j(t)\in \mathcal{G}$.} In Eq.~(\ref{eq:18}), the first term is the energy (or amplitude) component extracted from the base level halfwaves $\hslash^b_k(t)$ superimposed with the energy level of halfwaves at sub ($j<b$) and super ($j>b$) levels of $R^b{(t)}$. Here, the energy term in $\mathcal{I}(\hslash^b_{k}(t))$ is derived from the energy cascading principle, where energy is transferred from larger eddies (ocean currents) to smaller scale eddies as introduced by Kolmogorov and a similar inverse energy cascade principle\cite{paret1997experimental}. The superimposed energy component is then multiplied with the value of phase synchronization among the corresponding components. 

{\color{black}Distribution function of $\mathcal{I}(\hslash^b_{k}(t))\equiv {\sum_j g(\mathcal{E}_j)}\phi$, considering two arbitrary levels $j = 1,2$, can be expressed as the following product distribution:
\begin{equation}\label{eq:19}
F_{\mathcal{I}}(\iota) \propto\int_{-\infty}^{\infty}\frac{1}{2}f_{\sum g(\mathcal{E})}\left(\frac{\mathcal{\iota}}{\phi}\right)\frac{1}{|\phi|}d\phi
\end{equation} where $f(\cdot)$ is the density function of the energy term $\sum_{j=1,2} g(\mathcal{E}_j)$ that follows a generalized Pareto (GP) distribution with scale, shape and location parameters given as, $(c_1s_1+c_2s_2)/(s_1+s_2)$, $(1/s_1+1/s_2)$ and 0, respectively and the phase term, $\phi\sim U(-1,1)$. $\{c_1, c_2\}, \{s_1, s_2\}$ denote the scale and shape parameter of the GP distributions representing $g(\mathcal{E}_j),j=1,2$. Please see Appendix I of the supplementary material for details on the derivation of the distribution function and parameter estimation. The distribution can be similarly generalized for $j>2$.}

%

Notionally, InSync is analogous to the energy-based statistics employed in multi-scale analysis methods for change detection, fused with the intrinsic phase synchronization component. With an additional phase synchronization component, the statistic can capture the dynamic as well as sharp change-related information contained in various signal components more effectively compared to other contemporary methods as can be gathered from various case studies presented in the following section.

\section{Case Studies}
\subsection{Experimental setup}
 
\changeR{We investigate the performance of the InSync statistic for detecting changes in two numerical simulations and six real-world case studies under different nonlinear and nonstationary conditions. The numerical simulations include (a) detection of sharp changes in the dynamics of the logistic map and (b) changes in the noise variance structure of a piece-wise stationary ARMA(2,1) process. For the real-world case studies, we analyze data from healthcare and manufacturing systems. These include (a) detection of machine breakdown, (b) spike trains in neocortical signals, (c) onset of obstructive sleep apnea, (d) sharp transition in the spectral content of vibration signals, (e) detection of eye blinking events, and (f) trend shift in the Nile flow rate. For brevity of the results, we only report the first two real-world case studies here. A summary of the remaining case studies is presented in Table~\ref{table:t5}. For each case study, we compare the performance of our method with EWMA, Wavelet based CUSUM (WCUSUM), DPGSM, Pruned Exact Linear Time (PELT) from the \textit{CPM} package, and likelihood ratio test (LRT) from the \textit{changepoint} package. }

\change{We employ ARL to compare the performance of each method. It measures the total number of in-control or out of control data points that needs to be observed, on an average, before an anomaly can be detected. Therefore, if the process is in-control, a higher value of ARL (denoted as ARL0) indicates a lower likelihood of observing an anomaly and vice-versa. In contrast, when the process is out of control, a low ARL value (denoted as ARL1) is desirable. The control chart limit (i.e., the threshold on $\mathcal{I}^b(t)$ for declaring a change point) is determined by targeting a specific ARL0 value given as, $\text{ARL0}= {1}/\alpha$ where $\alpha$ is the Type I error rate when the process is in-control. For comparison purposes, we consider $3\sigma$ control limits, i.e.,  $\alpha = 0.0027$ such that the value of ARL0 is approximately equal to 370 sample points \cite{wang2018dirichlet}. ARL1 values are then estimated using the CUSUM chart of the InSync statistic as $\text{ARL1}= {1}/(1-\beta)$ where $\beta$ is the probability of type II error when the process is out of control. For a good change detection algorithm, we expect higher values of ARL0 and lower values of ARL1.} We generate 100 replications of the time series in numerical simulations to develop a consistent estimate of ARL1. We also report the true positive (TP) and false positive (FP) rates when it was not possible to estimate the ARL1 values for the competing methods.

\begin{figure}
	\includegraphics[width = 0.32\textwidth]{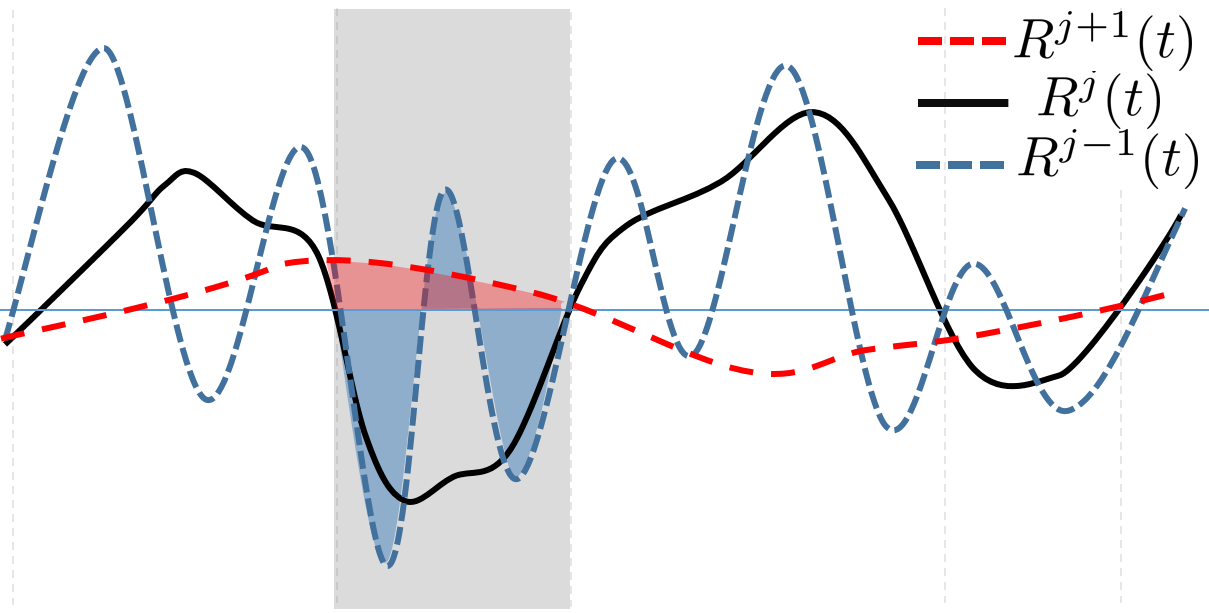}
	\centering
	\caption{ Illustrative example to show the halfwave span property of the rotation components, $R^j(t)$. We notice that the support of $\hslash^j_k(t)$, i.e., $(z_k^{j},z_{k+1}^{j}]$ in $R^j(t)$ spans 3 halfwaves from the previous level, $R^{j-1}(t)$ and a fraction of halfwave from the next level rotation component, $R^{j+1}(t)$.}
	\label{fig:im6}	
\end{figure}

\subsection{Recurrence plot based change point visualization}
Along with the ARL1 values, we also employ recurrence plots (RP) to visualize the change points. RP is a non-linear time series analysis tool which provides a two dimensional representation $[\mathbb{D}]_{ij}=||x^m(t_i)-x^m(t_j)||;i,j\in n$ of the evolution of the trajectory of the time series in the phase space. Here, $x^m(t_i)$ is the realization of the trajectory at time $t_i$ when embedded in an $m$-dimensional phase space such that: \begin{equation}{\label{eq:20}}
x^m(t_i)=\left(x(t_i),x(t_{i+d}),x(t_{i+2d}),...,x(t_{i+(m-1)d})\right)
\end{equation} where $m$ and $d$ are the optimal embedding dimension and time delay, respectively \cite{cheng2015time}. Due to Taken (Taken's theorem \cite{cheng2015time}), \change{$x^m(t_i)$ and the underlying true trajectory of $x(t)$ in the state space are diffeomorphisms}, hence representing the same dynamical system, but in different coordinate systems.  


\begin{figure}[t] 
	\includegraphics[width = 0.48\textwidth]{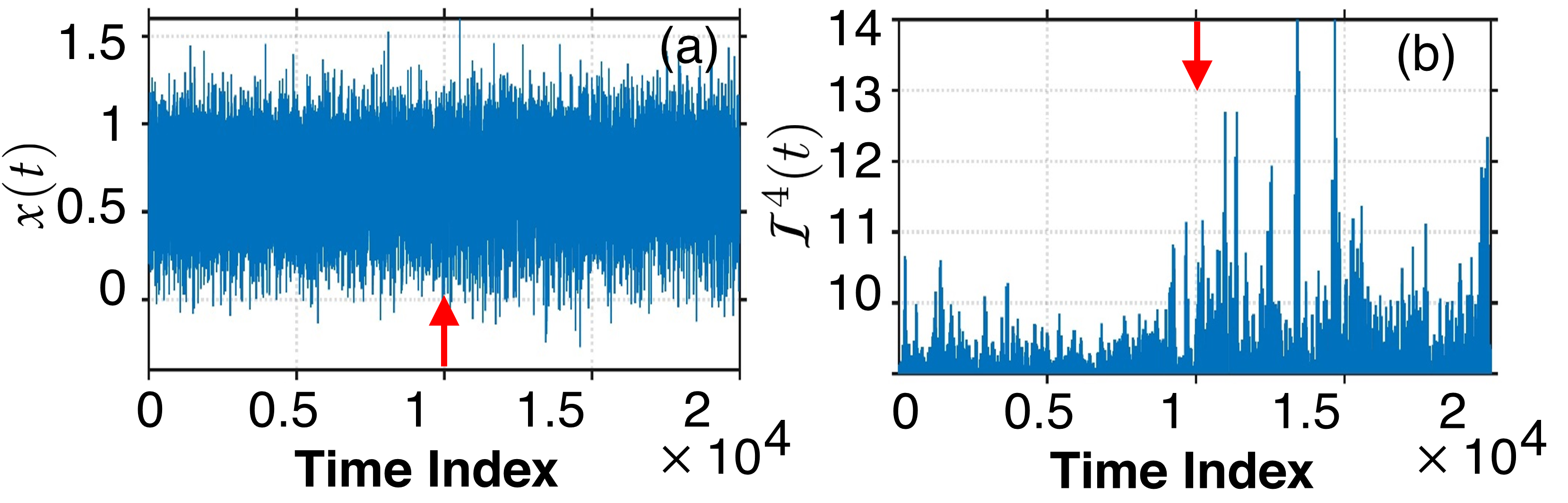}
	\centering
	\caption{(a)  Time portrait of logistic map $x(t)$ with the SNR $=10$~dB. Here, the change point is indicated by the arrow at 10000 t.u.; (b)  shows the InSync statistic $\mathcal{I}(\hslash^4_{k}(t))$ with the set $\mathcal{G}$ being $\{R^3(t), R^4(t), R^5(t)\}$.}
	\label{fig:im7}
\end{figure}

\begin{table}[!b]
\small
	\renewcommand{\arraystretch}{1}
	\caption{Comparison of ARL1 for different values of SNR (dB) in Logistic maps}
	\label{table_example}
	\centering	
	\begin{threeparttable}
		\begin{tabularx}{0.48\textwidth}{c *{6}{Y}}
			\hline
			{\footnotesize SNR} & {\footnotesize EWMA} & {\footnotesize WCUSUM} & {\footnotesize PELT\tnote{\#}} & {\footnotesize LRT} & {\footnotesize DPGSM} & {\footnotesize InSync}\\ 
			\hline
			20 & 36.15 & 8.49  & 1.01 & 1.17 & 1.01 & 1.23\\
			
			10 & 205.35 & 175.6\tnote{$^{\dagger}$}  & 1.09 & 1.94 & Inf & 1.32 \\
			
			6.67 & 288.43 & 199.68\tnote{$^{\dagger}$}  & 1.04 & 14.85 & Inf & 1.46\\ 
			
			5  & 317.46 & 277.49\tnote{$^{\dagger}$}  & 1.05 & 38.76 & Inf & 1.88\\ 
			\hline
		\end{tabularx}
		\begin{tablenotes}
			\item[]\footnotesize{{$^{\dagger}$Failed to detect any change in 10\% of the total runs}}
			\item[]\footnotesize{{$^{\#}$Failed to detect any change in $\geq$40\% of the total runs}}
		\end{tablenotes}  
	\end{threeparttable}
	\label{table:t1}
\end{table}

\subsection{Dynamic regime change in the logistic map}
To test the performance of our method for detecting changes between two nonlinear regimes, we generate a time series $x(t)$ with 20000 data points from the following logistic map model, superimposed with Gaussian noise: 
\begin{equation}{\label{eq:21}}
\begin{aligned}
x(t) &= y(t) + \mathcal{N}(0,\sigma^2) \\     y(t) &=\mu y(t-1)(1-y(t-1));\mu>0,t\in\mathbb{Z}^+ 
\end{aligned}
\end{equation}

\noindent The value of the signal to noise ratio (SNR) is varied from 20 to 5 by changing the value of $\sigma$ in Eq.~(\ref{eq:21}). The SNR is calculated as $\text{SNR}= 10\log_{10}(P_{\text{signal}}/P_{\text{noise}})$ where $P$ is the average power. A typical realization of $x(t)$ with SNR $=10$~dB is shown in Fig.~\ref{fig:im7}(a).

For the in-control regime, we set the value of $\mu$ in Eq.~(\ref{eq:21}) to 3.4 such that the logistic map exhibits a periodic behavior. A change is introduced in the dynamics of the process at $t = 10000$ time units (t.u.) by changing the value of $\mu$ to 3.7 where it exhibits a chaotic behavior. Clearly, this change is not discernible from the direct examination of the time portrait of the process as shown in Fig.~\ref{fig:im7}(a). To implement the proposed methodology, we first determine the set of rotation components with maximum mutual agreement. We note that $\mathcal{G} = \{R^3(t), R^4(t), R^5(t)\}$ with $R^4(t)$ as the base component. The corresponding InSync statistic $\mathcal{I}(\hslash^4_{k}(t))$ is shown in Fig.~\ref{fig:im7}(b). \changeR{One can note a discernible contrast in the values of the statistic between the two dynamic regimes. We assess the performance of the proposed method by comparing the ARL1 for different SNR values and are shown in Table~\ref{table:t1}. For the highest SNR value of 20 dB, PELT, LRT, and DPGSM perform relatively better than InSync. However, for lower SNR values (10 and below), we note that the ARL1 for the InSync statistic is almost two orders of magnitude smaller than EWMA, WCUSUM, and LRT. We exclude DPGSM from comparison as it has an ARL1 value of infinity, i.e., it failed to detect any change. Also, note that even though PELT performed better than the InSync statistic, it failed to detect the change point in more than 40\% of the simulation runs.}

\begin{figure}[!b]
	\includegraphics[width = 0.5\textwidth]{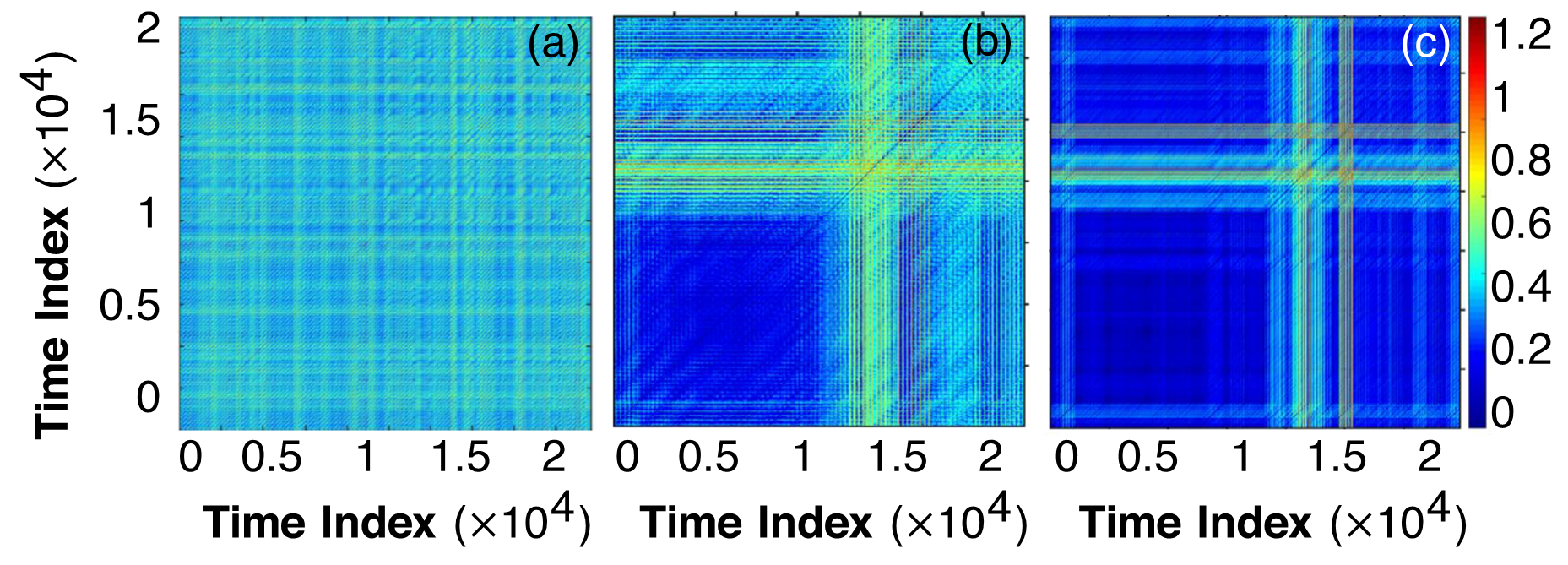}
	\centering
	\caption{ RP of the (a) logistic map $x(t)$ as given in Eq.~(\ref{eq:21}), (b) the InSync statistic $\mathcal{I}(\hslash^4_{k}(t))$ calculated using the set of rotation components $\mathcal{G}=\{R^3(t), R^4(t), R^5(t)\}$, and (c) the statistic calculated using only the base component $R^4(t)$.}	
	\label{fig:im8}
\end{figure}

We now shed some light on the contrast enhancement capability of the InSync statistic by using RP. To construct the RP, we used an embedding dimension $m=3$ and time delay $d=10$ (see \cite{yang2011local} for details on the calculation of $m$ and $d$). As expected, the RP of the original time series $x(t)$ shown in Fig.~\ref{fig:im8}(a) does not show any discernible change in the dynamics of the process. However, there is an apparent contrast in the RP constructed from the InSync statistic, as shown in Fig.~\ref{fig:im8}(b), capturing the change in the dynamics~of~the~process. 

We also establish the significance of the intrinsic phase and amplitude synchronization, i.e., fusing the phase and amplitude information across multiple components. For this, we examine the RP constructed just from a single component, say the base component (here $R^4(t)$) as shown in Fig.~\ref{fig:im8}(c). Although the RP based on $R^4(t)$ contains some information of the change point, it is not able to differentiate the regimes as effectively as the RP of $\mathcal{I}(\hslash^4_{k}(t))$ does. This is because a single rotation component contains only a fraction of the information associated with the change point. This supports our assertion that the information pertaining to the change point is reinforced whenever the phase and amplitude information across the set of rotation components are fused, thereby increasing the specificity and sensitivity of detecting the change points.

\begin{figure}[!t]
	\includegraphics[width = 0.32\textwidth]{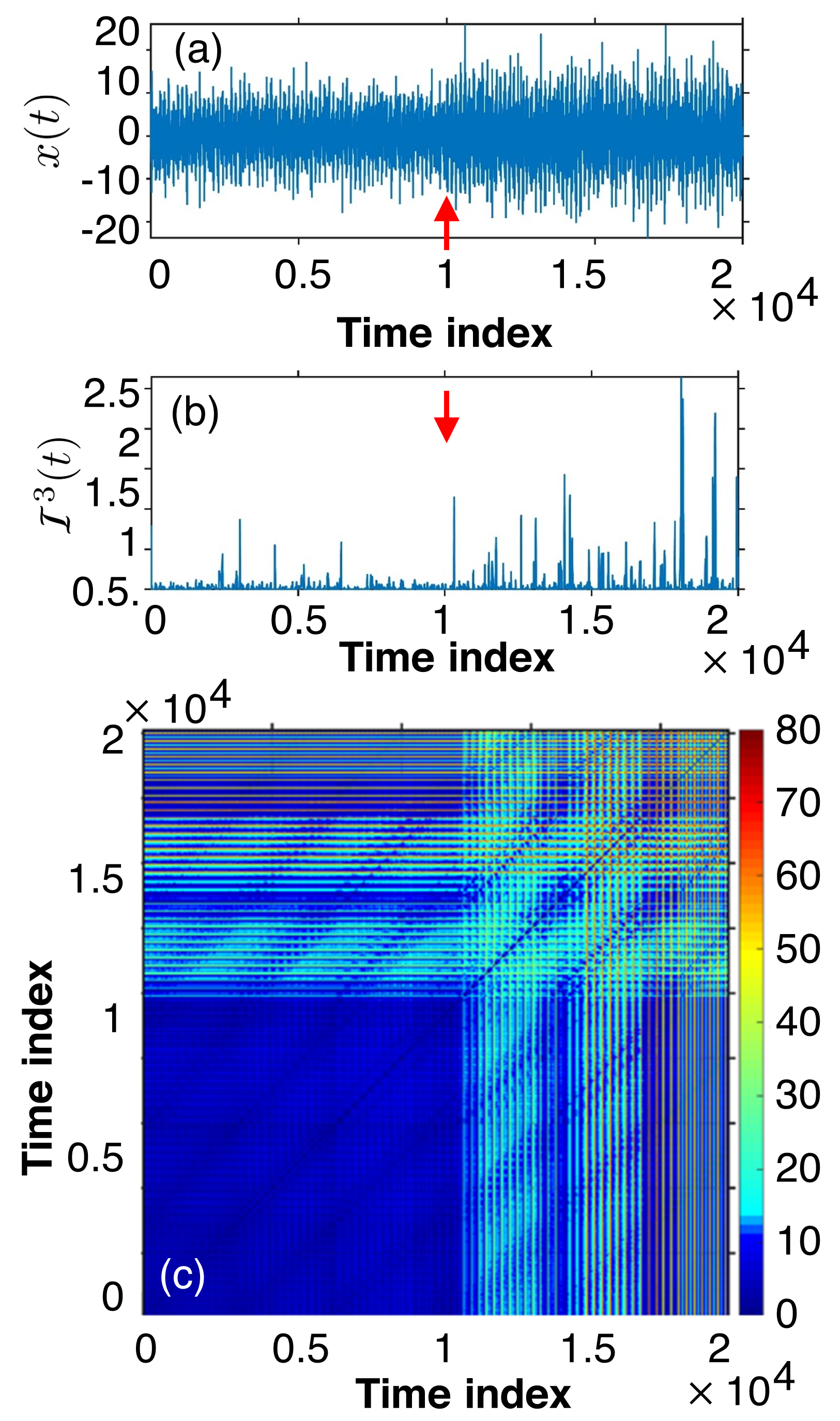}
	\centering
	\caption{(a) Piecewise ARMA time series, (b) InSync statistic $\mathcal{I}(\hslash^3_{k}(t))$ calculated from the rotation components $\{R^2(t), R^3(t), R^4(t)\}$ and (c) shows the corresponding RP.}
	\label{fig:im9}
\end{figure}

\subsection{Piecewise stationary ARMA (2,1)}
In the second simulation study, we test the performance of the method in detecting changes between two nonstationary regimes exhibiting intermittency. For this, we generate a time series with 20000 data points from a piecewise stationary ARMA(2,1) process with \textit{M} breakpoints given as:
\begin{equation}{\label{eq:22}}
\Xi^{(m)}(\varrho)x(t)=\varOmega^{(m)}(\varrho)w(t); t\in\mathbb{Z}; t_{m-1}\leq t<t_m
\end{equation} where $\Xi^{(m)}$ and $\varOmega^{(m)}$ are polynomials of degree 2 and 1, respectively, $\varrho$ is the lag operator and $t_m, m=1,2,\ldots,M$ is the time index of each breakpoint such that $t_0=1$ and $t_M=T$. The sojourn time $t_m-t_{m-1}, m=1,2,...,M$ for each breakpoint is fixed to 100 data points. When the system is in in-control state, the noise variance shock (NVS) follows a normal distribution, i.e., $w(t)\sim N (0,\delta \sigma^2)$ where the variance multiplier cycles as $\delta=\left\lbrace 1,3,5,1,3,5,\ldots\right\rbrace $. Change in the system is introduced at $t = 10000$ t.u. by changing the amplitude of NVS. Figure~\ref{fig:im9}(a) shows a  representative ARMA(2,1) time series where the change point is introduced by changing the NVS from NVS$_0=\{1,3,5,1,3,5,...\}$ to NVS$_a=\{1,4,5,1,4,5,...\}$. Other change points that were assessed include $\left\lbrace \underline{1,3,5, \ldots} \right\rbrace  \rightarrow\left\lbrace  \underline{ 1,4,6, \ldots} \right\rbrace$ and $\left\lbrace \underline{1,3,5, \ldots} \right\rbrace  \rightarrow\left\lbrace  \underline{ 2,4,6, \ldots} \right\rbrace$.  Note that the intermittent behavior of the system makes it difficult to detect such changes~\cite{wang2018dirichlet}. 

We now employ the InSync statistic to detect the changes in the NVS. Using the concept of mutual agreement, we obtain $R^3(t)$ as the base component and $\mathcal{G} = \{R^2(t), R^3(t), R^4(t)\}$ as the set of ITD components with maximum mutual agreement. The corresponding InSync statistic and the RP constructed from the from InSync statistics are shown in Figs.~\ref{fig:im9}(b) \& (c). \change{Visually, the RP is able to capture the transition in the dynamics after the change has occurred}. Table~\ref{table:t2} reports the performance of the InSync statistic in terms of the ARL1 values for different levels of change in the NVS. \changeR{While PELT and DPGSM performs marginally better when the NVS is changed to $\left\lbrace  \underline{ 2,4,6} \right\rbrace$, InSync statistic detects the change much earlier than all the other methods as the magnitude of the change in NVS decreases. We also note that PELT failed to detect the change point in 20\% of the runs.}

\begin{table}[!t]
\small
	\caption{ARL1 for different values of NVS$_a$ in piecewise stationary ARMA(2,1) with NVS$_0 =\left\lbrace \underline{ 1,3,5} \right\rbrace $ }
	\label{table_example}
	\centering
	\begin{threeparttable}
		\begin{tabularx}{0.48\textwidth}{bbbbsbb}
			\hline
			 {\footnotesize NVS$_a$} & {\footnotesize EWMA} & {\footnotesize WCUSUM} & {\footnotesize PELT} & {\footnotesize LRT} & {\footnotesize DPGSM} & {\footnotesize InSync}\\ 
			\hline
			$\left\lbrace  \underline{ 2,4,6} \right\rbrace $ & 109.4 & 6.87 & 1.01 & Inf & 1.01 & 1.06\\
			\hline
			$\left\lbrace  \underline{ 1,4,6} \right\rbrace $ & 113.04 & 14.28 & 15.71 & Inf & 1.35 & 1.07\\
			\hline 
			$\left\lbrace \underline{ 1,4,5} \right\rbrace $ & 329 & 137.8\tnote{\#} & 49.24 & Inf & 1.77 & 1.28\\
			\hline
		\end{tabularx}
		\begin{tablenotes}
			\item[]\footnotesize{{$^{\#}$Failed to detect change in 20\% of the runs}}
		\end{tablenotes}
	\end{threeparttable}
	\label{table:t2}
\end{table}

\begin{figure}[!b]
	\includegraphics[width = 0.35\textwidth]{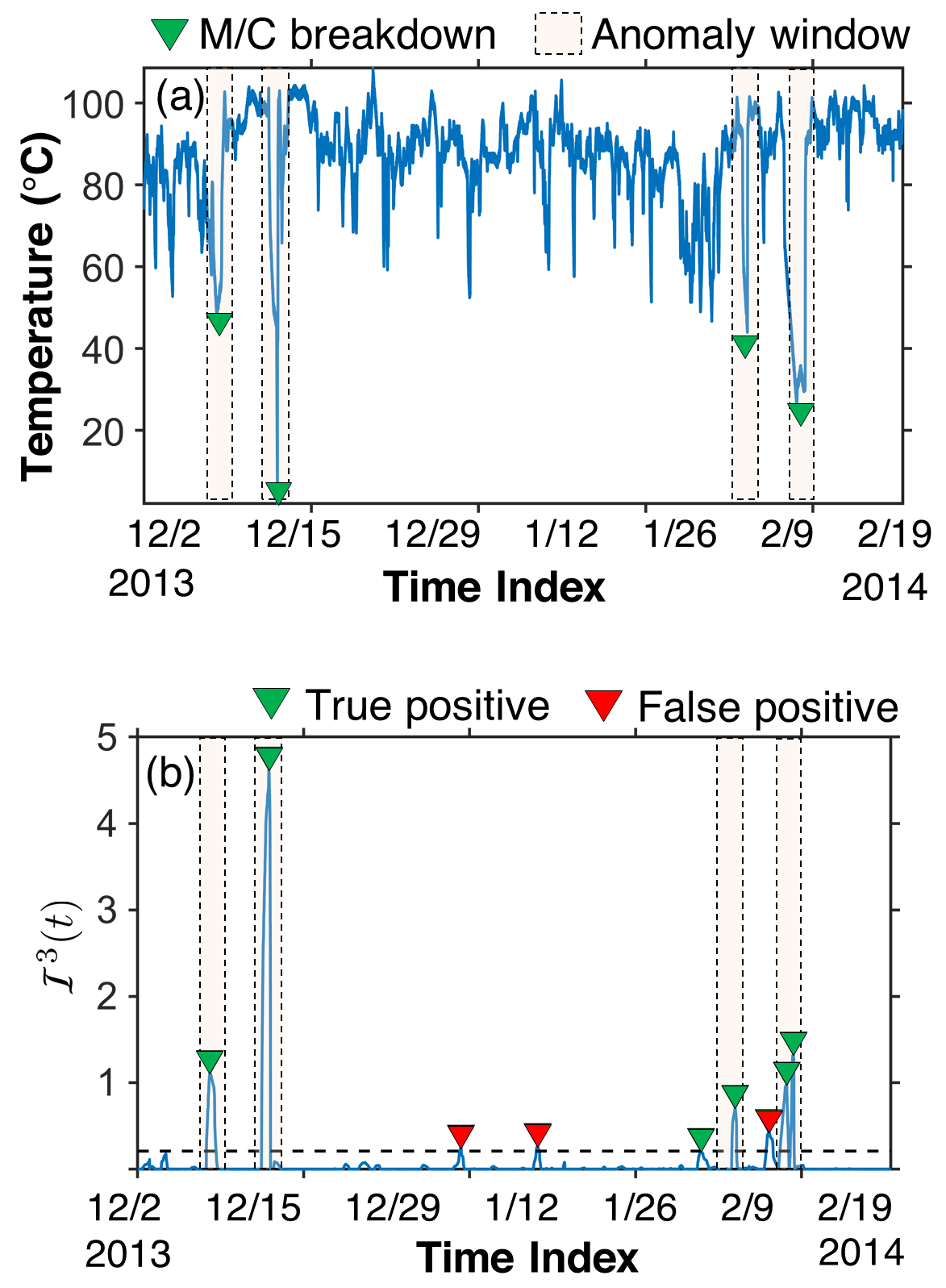}
	\centering
	\caption{(a) Temperature sensor measurement, recorded every 5 min. The machine breakdowns are shown with a green triangle and the anomaly window is appropriately highlighted. (b) The corresponding InSync statistic $\mathcal{I}(\hslash^3_{k}(t))$ calculated from the rotation components, $\{R^2(t),R^3(t),R^4(t)\}$. Singularities that lie within the anomaly window are deemed as TP rest as FP.}
	\label{fig:im10}
\end{figure}

\subsection{Industrial anomaly detection}
To establish the effectiveness of the proposed methodology on real-world nonstationary systems, we first examine the benchmark dataset on anomaly detection in an industrial machine \cite{ahmad2017unsupervised}. The dataset contains temperature measurements of an internal component of the machine recorded every 5 min for 79 consecutive days, as shown in Fig.~\ref{fig:im10}(a). The machine breaks down whenever the temperature of the component abruptly goes below a specified limit. Since the machine breakdowns are associated with abrupt changes that are short-lived, we consider the corresponding change points as singularities in the system. 

\changeR{In this case study, we subscribe to a standard scoring function proposed in \cite{ahmad2017unsupervised}. Using an anomaly window, the scoring function assigns a positive score for a true detection and penalizes for any missing anomalies or false positives. The anomaly window is set to 10\% of the length of the time series, divided by the total number of anomalies in the dataset. See Appendix J of the supplementary material for details.}

To detect these singularities, we analyze the set $\mathcal{G}=\{R^2(t),R^3(t),R^4(t)\}$ of rotation components exhibiting maximum mutual agreement with $R^3(t)$ as the base component. \changeR{The InSync statistic calculated using these components is shown in Fig.~\ref{fig:im10}(b). To minimize the false alarm rate, we determine the threshold on $\mathcal{I}(\hslash^3_{k}(t))$ by setting ARL$0 = 10^{4}$ samples for the in-control region and is shown in Fig.~\ref{fig:im10}(b) with a dotted black line. With this threshold, we note that eight singularity points are detected by the InSync statistic. Per the benchmark scoring function, if more than one anomaly is detected within an anomaly window for a given change point, only the first detection point is considered and the rest are ignored. Therefore, the second anomaly detected in the last anomaly window is ignored for calculating the final score. We also note that out of these detected singularities, four lie outside any anomaly window. These represent the false positives and are marked with red. Anomalies that are correctly identified (true positives) are marked with green. To compare the performance of our method with the benchmark methods (see \cite{ahmad2017unsupervised} for details), we use the TP and FP along with the scoring function. The results are summarized in Table~\ref{table:t3}. Among all the methods, the InSync statistic was able to detect the anomalies with the least number of false alarms and achieves the highest score. The ARL1 values corresponding to EWMA, WCUSUM, PELT, DPGSM, and InSync. are presented in Table~\ref{table:t5}}.

\begin{table}[!t]
\small
	\renewcommand{\arraystretch}{1.05}
	\caption{Comparison of the TP, FP and NAB benchmark score (out of 4) for various benchmark methods. }
	\label{table_example}
	\centering
	\begin{threeparttable}
		\begin{tabular}{c c c c}
			\hline
			Methods & TP & FP & Score \\ 
			\hline 
			\textbf{InSync} (ARL$0 = 10^{4}$) & \textbf{4} &  \textbf{1} & \textbf{3.89}\\
			PELT (ARL$0 = 10^{4}$) & 4 & 2 & 3.78\\
			Hierarchical Temporal Memory$^{\dagger}$ & 4 & 12 & 2.68\\
			Contextual Anomaly Detector$^{\dagger}$& 2 & 0 & -0.165\\
			Relative Entropy$^{\dagger}$& 2 & 9 &  -0.916\\		
			KNN-CAD$^{\dagger}$ &  2 & 20 & -2.856\\
			Bayesian Change point$^{\dagger}$& 1 & 4 &  -3.320\\		
			\hline 		
		\end{tabular}
	\begin{tablenotes}
		\item[]\footnotesize{$^{\dagger}$See \cite{ahmad2017unsupervised} for reference}
	\end{tablenotes}
	\end{threeparttable}
	\label{table:t3}
\end{table}
\subsection{Detecting singularities in neocortical signal patterns}
In this case study, we demonstrate the efficacy of the proposed method in detecting singularities in the signal recorded from the neocortex region of the brain. Reliable and accurate detection of these spikes is still an open problem because (a) the spikes are oftentimes mistaken with other electrical activities \cite{gordon2013comparing} or (b) the waiting time between two spikes may be as low as 0.6 ms, which might not be resolved by conventional change detection methods. Data used in this study was collected for 60 s at a sampling rate of 24~kHz  \cite{quiroga2004unsupervised}. A 50 ms realization of the signal is presented in Fig.~\ref{fig:im11}(a). 

Since the problem involves detecting singularities, we expect the halfwaves containing these singularities to show a strong phase synchronization at multiple levels of $R^j{(t)}\in \mathcal{G}$ (see Proposition 2). To detect these singularities, we identify $R^1(t)$ and $R^2(t)$ as the set of rotation components with maximum mutual agreement. This is apparent since the singularities are high frequency features and the change point information pertaining to singularities is resolved mostly in the first few levels. The corresponding plot of the InSync statistic and the RP constructed from the InSync statistic are shown in Figs.~\ref{fig:im11}(b) \& (c), respectively. The singularity points (represented by sharp vertical lines) can be easily visualized from Fig.~\ref{fig:im11}(c). 

\begin{figure}[!t]
	\includegraphics[width = 0.38\textwidth]{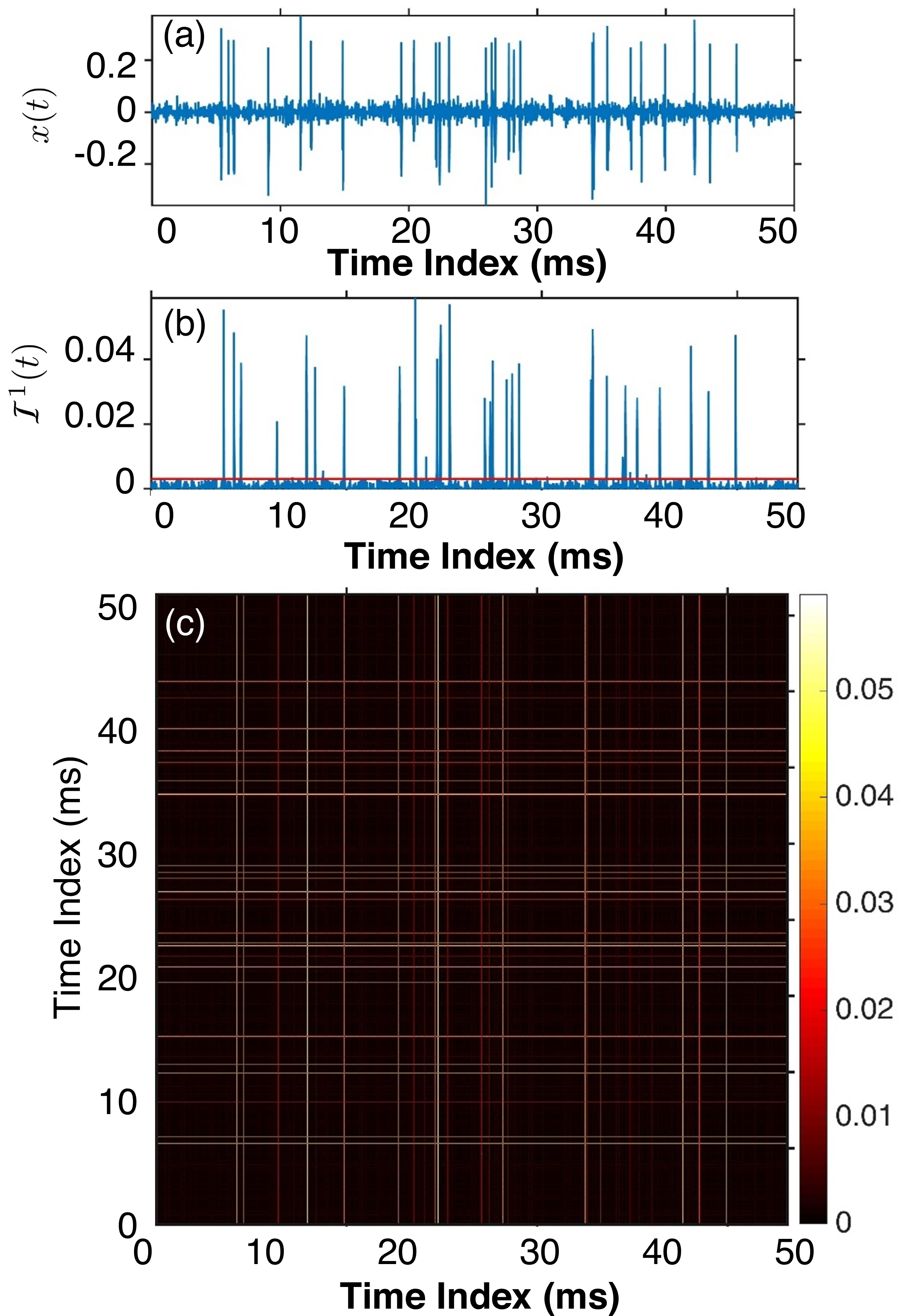}
	\centering
	\caption{(a) A 50 ms realization of the neocortical recording (b) the corresponding plot of the InSync statistic, $\mathcal{I}(\hslash^1_{k}(t))$ calculated from the rotation components, $R^1(t)$ and $R^2(t)$ (c) shows the RP constructed from the~time~series~of~$\mathcal{I}(\hslash^1_{k}(t))$.}
	\label{fig:im11}
\end{figure}

\begin{table}[!b]
\small
	\renewcommand{\arraystretch}{1.05}
	\caption{Comparison of FN for different noise levels (NL)}
	\label{table_example}
	\centering
	\begin{tabular}{c c c c c c}
		\hline
		NL\tnote{\#} & Spike Count &  SPC  & WCUSUM &  PELT & InSync\\ 
		\hline 
		0.05 & 3514   & 194 & 33 & 549 & 52 \\
		0.10 & 3448 & 178  & 38 & 694 & 98 \\
		0.15 & 3472  & 258 & 59 & 1152 & 176 \\
		0.20 & 3414 & 626  & 68 & 1747 & 255 \\
		\hline
	\end{tabular}
	\begin{tablenotes}
			\item[]\footnotesize{{$^{\#}$NL represents the signal standard deviation relative to the peak amplitude of the spikes.}}
		\end{tablenotes}  
	\label{table:t4}
\end{table}

\begin{sidewaystable*}
	\caption{Summary of various simulated and real world case studies implemented using the intrinsic phase and amplitude synchronization with corresponding average ARL1. ($^{\dagger}$not reported in Section 4, $^{\dagger\dagger}$see Appendix G in the supplementary material, $^{\#}$results in over segmentation in the presence of noise.)}
	\label{table_example}
	\centering
	\begin{threeparttable}
		\bgroup
		\def\arraystretch{1.1}%
		\begin{tabular}{>{\centering\bfseries}m{1in} >{\centering}m{0.8in} >{\centering}m{0.35in} >{\centering}m{0.45in} >{\centering}m{0.35in} >{\centering}m{0.350in} >{\centering}m{0.35in} >{\centering}m{0.35in} >{\centering\arraybackslash}m{4.0in}}
			\hline
			Case  study& Change point(s) & EWMA & Wavelet-CUSUM & PELT & LRT$^{\#}$ & DPGSM & InSync & 
			Methods implemented and remarks  \\ 
			&  & \\ \hline
			Logistic map & Periodic to chaotic  &211.84  & 165.32 & 1.05 & 14.18 & \textbf{1.01} &  1.47 & InSync consistently detected the change point with $\overline{\text{ARL1}}$ almost 2 orders of magnitude smaller as compared to EWMA and WCUSUM. PELT and DPGSM did not detect the change in at least 40\% of the simulation runs.\\	\hline
			Piecewise stationary ARMA(2,1) \cite{wang2014change} & Change in NVS  & 183.82 & 52.98 & 21.99
			& Inf & 1.37 & \textbf{1.13} &  InSync was able to detect the smallest change in NVS 	($\{ \underline{1,3,5} \} \rightarrow \{ \underline{ 1,4,5} \} $).\\ \hline
			Machine temperature sensor data \cite{ahmad2017unsupervised} & Machine breakdown states & 4 & 4 & 1 & Inf & Inf & \textbf{1} &  Overall score assigned to InSync was $\approx33\%$ higher than other benchmark methods tested. Overall score for EWMA and WCUSUM  were  -4.94 and -11.5, respectively.\\ \hline
			Neocortex spike detection \cite{quiroga2004unsupervised} & Neuronal firings & 1.25 & \textbf{1.01} & 2.05 & Inf & Inf & {1.05} & WCUSUM detected spikes with maximum sensitivity (0.98) followed by InSync(0.95). For SPC and EWMA, sensitivity values were, 0.93 and 0.80, respectively.\\ \hline
			Obstructive Sleep Apnea (OSA), \cite{le2013prediction}$^\dagger$ & No OSA to OSA  & 214 ms & 193 ms  & \textbf{10.13} ms & Inf & {11} ms & {12.76} ms &  Apnea event was detected fastest using PELT, followed by DPGSM and InSync with ${\overline{\text{ARL1}}}\approx 12.76$ ms.\\ \hline
			Chemical mechanical polishing \cite{wang2014change}$^\dagger$ &  Pad glazing after 9 mins of polishing & 90 ms & 17 ms & 2.22 ms & Inf & 9.42 ms  & \textbf{1.46} ms &  Pad glazing was detected by InSync with $\overline{\text{ARL1}}=1.46$ ms as compared to 2.22 ms \& 9.42 ms for PELT \& DPGSM.\\ \hline
			Eye blinking using EEG \cite{Lichman2013}$^\dagger$ & Eye opening and closing events (19 events)  & 1.35 & 1.19 & 2.3 & Inf & Inf & \textbf{1.05 }&  InSync detected the singularity events with a sensitivity value of 0.95 and 2 false positives. For EWMA and WCUSUM sensitivity values were recorded to be 0.74 with 19 false positives and 0.84 with 9 false positives, respectively. \\
			\hline 
			Nile flow rate$^{\dagger\dagger}$ & Trend shift& 2.98 & 300 & 1.01 & \textbf{1} & Inf & 1.01 & Demonstrates the efficacy of InSync in detecting trend changes with $\overline{\text{ARL1}}=1.01$. \\ \hline
			Logistic map $\bigodot$ Neocortex$^{\dagger\dagger}$   & Multiple change points & 163.6, 261.4, 1.2 & 1.41, 1.01, 1.01 & Inf & Inf & Inf & \textbf{1.11, 1.10, 1.01 } & Concatenation ($\bigodot$) of Logistic map and neocortex signal. Consists of dynamic pattern change, variance shift with interspersed singularities. InSync performed consistently in detecting all the change points.\\ \hline \hline 
			Computational cost (seconds) & Logistic map & 0.0028 & 0.77 &  0.9 & 0.3 & $>1$ & 0.02 & The InSync statistic has a lower computational cost as compared to all the comparative methods except for EWMA.   \\  \hline
		\end{tabular}
		\egroup
	\end{threeparttable}
	\label{table:t5}
\end{sidewaystable*}



We compare the performance of the InSync statistic against the superparamagnetic clustering (SPC) algorithm proposed in \cite{quiroga2004unsupervised}, WCUSUM, and PELT for \change{SNR values of 20, 15, 10, and 5 dB}. For comparison purposes, we use the number of false negatives (FN) as reported in Table ~\ref{table:t4}. The corresponding ARL1 values are summarized in Table~\ref{table:t5}. We notice that the InSync statistic and WCUSUM are able to detect the spikes in all the cases with a relatively higher sensitivity (lower FN) as compared to SPC and PELT. Since SPC is based on identifying the shape features of the spikes followed by clustering, it is highly likely that the shape features of the spikes may not belong to the same cluster. In contrast to SPC, InSync statistic enjoys the advantage that spikes are short-lived changes and are retained across multiple levels of rotation components with a very high probability as compared to random signal~fluctuations.

{\color{black}We also note that the number of FP for the InSync statistic increases as the SNR decreases for a fixed value of sensitivity. This may be attributed to the fact that rotation components tend to retain random fluctuations in a given signal across multiple decomposition levels---although with a very small probability---causing the fluctuations to appear as FP in the InSync statistic. \changeR{In contrast to intrinsic time scale decomposition, wavelet decomposition (used in WCUSUM) provides better temporal localization of high-frequency signal features \cite{frei2007intrinsic} leading to fewer false negatives as compared to InSync.} However, compared to the sample size ($1.44\times 10^6$), the influence of the number of FP on the overall performance (specificity) of the InSync statistic would be negligible.}

\section{Summary And Discussion}
We presented an approach for detecting changes in nonlinear and nonstationary systems based on tracking the local phase and amplitude synchronization among multiple intrinsic components of a univariate time series signal obtained via intrinsic time scale decomposition. We showed that the signatures of sharp change points such as singularities and moment shifts are preserved across multiple ITD components with a significantly high probability as compared to random signal signatures. This is significant as it offers the possibility to leverage the information contained across multiple ITD components to detect change points. In this direction, we used a phase synchronization measure to show that the change point information is reinforced (amplified by at least four times) when the phase and amplitude information across a set of ITD components are fused. Subsequently, we developed a network-based maximum mutual agreement approach to identify the set of ITD components that are most likely to retain the change point information and developed an InSync statistic that combines and reinforces the phase and amplitude information contained across these ITD components.

We implemented the InSync statistic to detect change points in two simulated and six real-world case studies in healthcare as well as manufacturing systems. We used ARL1 values to compare the performance of our method with other classical approaches including EWMA, WCUSUM, and DPGSM along with two benchmark change detection packages, \textit{CPM}, and \textit{changepoint}. A summary of the performance measure is presented in Table~\ref{table:t5}. These results suggest that our method was able to detect change points with ARL1 on an average of almost 62\% lower as compared to the other methods tested. In addition, the InSync statistic was able to detect these sharp change points with a relatively high sensitivity ($\approx0.91$; on an average 20\% higher compared to the other methods) and low false positive rates. The significant increase in the sensitivity of the method is attributed to its contrast enhancement property. As the statistic combines the phase and amplitude information from multiple ITD components, it reinforces and amplifies the contrast between different intermittent regimes and sharp change points. 

However, there are some limitations to the present approach. First, the method is limited in detecting only sharp change points. Gradual changes in the mean and higher order moments are detected with relatively higher ARL1 values. Second, due to the small non-zero probability of retaining random signal features, the specificity of the InSync statistic may be lower in the presence of high noise levels. Our ongoing efforts are focused on addressing these issues by analyzing the trend or baseline components of the given signal.

%
%
%
\section*{Acknowledgments}
This work was supported by the kind funding from the National Science Foundation, grant no. CMMI-1432914, CMMI-1437139, IIP-1543226, IIP-1355765 and ECCS-1547075.

\bibliographystyle{IEEEtran}
\bibliography{references_IEEETSP}

\clearpage

\setcounter{assumption}{0}
\setcounter{proposition}{0}
\setcounter{corollary}{0}
\setcounter{lemma}{0}
\setcounter{proof}{0}
\setcounter{figure}{0}   
\setcounter{equation}{0}     
\setcounter{theorem}{0}
\renewcommand\thefigure{S\arabic{figure}}    
\renewcommand\theequation{S\arabic{equation}}    
\newcounter{storeeqcounter}
\newcounter{tempeqcounter}    

\begin{figure*}
\centering
      \Huge{Change detection in complex dynamical systems using intrinsic phase and amplitude synchronization: Supplemental Material}\\
\end{figure*}
\appendices
\section{Proof of Corollary 1}\label{appendix:A}
  
\begin{corollary} The probability that an extremum in the rotation component at level $j$ of $x_k$ is retained as an extremum across the subsequent $\eta$ rotation components is approximately equal to $ 0.24^{\eta} $.
\end{corollary}

To show that the assumption holds reasonably well, we first refer to [1] where the authors numerically showed that the consecutive extrema points in $x_k$ (Eq.~(6) of the main document) evolve via a random decimation process, i.e., the extrema points are equally likely to be destroyed or retained in the next level, independent of its neighbors. This makes the consecutive $\Delta^j_k$ independent. To verify this assumption, we analyze the autocorrelation function of $\Delta^j_k$ for the rotation component at level $j=2$ of $x_k$ and is shown in Figure~\ref{S1}. We note that the value of the autocorrelation remains close to zero for lags greater than one. Although, there is a small, but non-zero autocorrelation at lag one, we ignore this small correlation to simplify the calculations in the later part. As we see in Figure 3 of the main document, the probability estimation based on this assumption (Eq.~(10) of the main document) closely captures the probability estimation as made from the Monte Carlo simulation of Eq.~(7) of the main document that doesn't assume independence.

Using this independence assumption, we present the following two lemma necessary for the subsequent analysis.

\begin{lemma}
For a sufficiently long time series (i.e., $N^j\rightarrow\infty$), $\Delta^j_k$ follows an exponential distribution. 
\end{lemma}
\begin{proof}{}\label{one}
	\normalfont	{For any level $j-1$, let the extrema locations be denoted as $\{\tau^{j-1}_1,\tau^{j-1}_2,\ldots,\tau^{j-1}_n\}$. As the extrema points evolve via a random decimation process, the number of successive extrema points (say $c^j_{k}$) that disappear until an extremum is retained is geometrically distributed.	From the law of large numbers (LLN), the sample mean of inter-extremal separations at level $j-1$ within the interval $[\tau_{k-1}^j, \tau_k^j]$, i.e.,  $\mathbb{E}({\Delta_{k}^{j-1},\forall k :\tau^{j}_{k-1}<\tau^{j-1}_k<\tau_k^j})$ converges to the population mean $\mathbb{E}(\Delta_{k}^{j-1},\forall k = 1,\ldots, N^{j-1})$. From ITD we have, 		
	$$ \Delta_{k}^{j}  = {c^j_{k}}\mathbb{E}\left({\Delta_{k}^{j-1},\forall k :\tau^{j}_{k-1}<\tau^{j-1}_k<\tau_k^j}\right) $$ Since, $\mathbb{E}({\Delta_{k}^{j-1},\forall k :\tau^{j}_{k-1}<\tau^{j-1}_k<\tau_k^j})$ converges to $\mathbb{E}(\Delta^{j-1}_k) $, we note that $\Delta_{k}^{j}$ converges to $c^j_k\mathbb{E}(\Delta^{j-1}_k) $. Therefore, $\Delta^j_k$ is geometrically distributed with some parameter $p^j$. 

	Let $N^jp^j=\lambda^j$ where $N^j$ is the number of extrema in level $j$. Then in the limit $N^j\rightarrow \infty$, we have
	\begin{eqnarray}\label{Eq2}
	F_{\Delta^j_k}(\varDelta)&=&\lim_{N^j\rightarrow\infty} \sum_{i=0}^{\varDelta}{\left(1-\frac{\lambda^j}{N^j}\right)^i\frac{\lambda^j}{N^j}}\nonumber\\
	&=&\int_{0}^{\varDelta}\lambda^j\exp(-\lambda^j\Xi)d\Xi
	\end{eqnarray}
	\change{where  $\lambda^j=1/\mathbb{E}[\Delta^j_k]$.} \QED}
\end{proof}

\begin{figure}[!t]
\centering
\includegraphics[width = 0.40\textwidth]{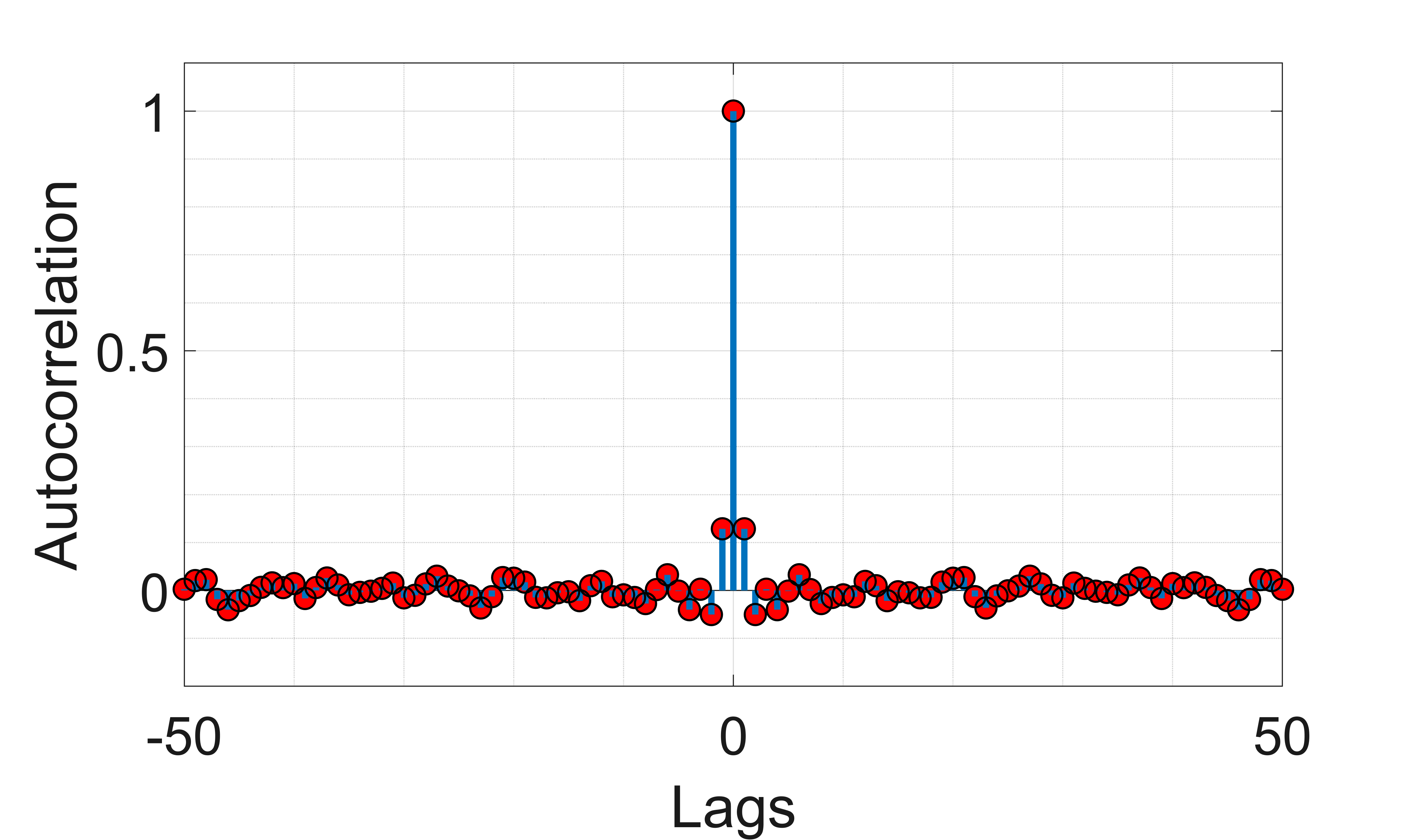} 
\caption{Autocorrelation of the inter-extremal separations}\label{S1}
\end{figure}

\begin{lemma}
	Let $q^j_k$ be defined as the ratio of the difference  to the sum of the inter-extremal separations given as: \begin{eqnarray}\label{Eq3}
	q^j_k := \frac{(\tau^j_k-\tau^j_{k-1})-(\tau^j_{k+1}-\tau^j_{k})}{(\tau^j_k-\tau^j_{k-1})+(\tau^j_{k+1}-\tau^j_{k})} =\frac{\Delta^j_k-\Delta^j_{k+1}}{\Delta^j_k+\Delta^j_{k+1}}
	\end{eqnarray} $\forall k=1,2,\ldots,N^j;j = 1,\ldots, J$. We show that $q^j_k$ follows a uniform(-1,1) distribution. 
\end{lemma}

\begin{proof}{}\label{one}
	\normalfont {Rewriting $q^j_k$ as,
		\begin{eqnarray*}
			q^j_k=\frac{\Delta^j_k-\Delta^j_{k+1}}{\Delta^j_k+\Delta^j_{k+1}} = \frac{\Delta^j_k/\Delta^j_{k+1}-1}{\Delta^j_k/\Delta^j_{k+1}+1}\\
		\end{eqnarray*} 
		Let $V=\Delta^j_k/\Delta^j_{k+1}$. The probability density of ratio of two exponential random variables, $\Delta^j_k$ and $\Delta^j_{k+1}$ can be derived as, 
		\begin{eqnarray*}
			f_{V}(v)&=&\int_{0}^{\infty}\rho^j_{k+1}f_{\Delta^j_{k}\Delta^j_{k+1}}\left(v\rho^j_{k+1},\rho^j_{k+1}\right) d\rho^j_{k+1}\\
			&=&\int_{0}^{\infty}\rho^j_{k+1}(\lambda^j)^2 e^{-(\lambda^jv\rho^j_{k+1})}e^{-(\lambda^j\rho^j_{k+1})}d\rho^j_{k+1}\\
			&=&\int_{0}^{\infty}(\lambda^j)^2\rho^j_{k+1} e^{-\lambda^j\rho^j_{k+1}(1+v)}d\rho^j_{k+1}\\
			&=&\frac{1}{(1+v)^2}
		\end{eqnarray*}
		where, $v\in(0,\infty)$. Since, $q^j_k=(V-1)/(V+1) \implies  q^j_k\in(-1,1)$. Using change of variables, we have the density function for $q^j_k$ given as, 
		\begin{eqnarray*}
			f(q^j_k)&=&\frac{1}{(1+v)^2}\bigg|\frac{dV}{dq^j_k}\bigg|\\
			&=&\frac{1}{\left(1+\left( \frac{1+q^j_k}{1-q^j_k}\right) \right)^2}\times\frac{2}{\left(1-q^j_k\right)^2}\\
			&=&\frac{1}{2}
		\end{eqnarray*}
		Since $q^j_k\in(-1,1)$ and $f(q^j_k)=\dfrac{1}{2}; q^j_k\sim$ uniform(-1,1). \QED}
\end{proof}

Using the above results, we now present the proof of Corollary 1. Here, we only provide an outline of the proof. Please refer to \cite{restrepo2014defining} for a detailed proof and calculations involved.
\begin{proof}
\normalfont{Based on \cite{restrepo2014defining}, any three consecutive extrema in level $j+1$, say, $r^{j+1}_1, r^{j+1}_2$ and $r^{j+1}_3$, given the corresponding realizations of $q^j_1, q^j_2$ and $q^j_3$ in level $j$, follows a joint Gaussian distribution with joint conditional density given as:
\begin{multline}\label{Eq4}
f\left( r^{j+1}_1, r^{j+1}_2,r^{j+1}_3\big|q^j_1, q^j_2,q^j_3\right)   \\ =  \frac{1}{\sqrt{8\pi^3\text{Det}\left(\sum\left( q^j_1, q^j_2,q^j_3\right)  \right) }} e^{\left( -\frac{1}{2}\textit{\textbf{r}}^\textit{\textbf{T}} \sum\left( q^j_1, q^j_2,q^j_3\right)\textit{\textbf{r}}\right)}
\end{multline} 
where, $\textit{\textbf{r}} =\left\lbrace r^{j+1}_1, r^{j+1}_2,r^{j+1}_3\right\rbrace $ with covariance matrix expressed as follows: 
\begin{multline}\label{Eq5}
\sum\left( q^j_1, q^j_2,q^j_3\right)  = MM^T
= \\ \left[
\begin{matrix}
6+2 q^2_1 &  4+2q_1-2q_2  & (1+ q_1)( 1-q_3)\\
4+2q_1-2q_2 &  6+ 2q_2  & 4+2q^j_2-2q_3 \\
(1+ q_1)( 1-q_3) &  4+2q_2-2q_3   & 6+2q_3^2
\end{matrix}\right]
\end{multline}
Now, given the exponential distribution of inter-extremal separations $\Delta^j_k$ (see Lemma 1) and the expression for $q^j_k$ as given in Eq.~(\ref{Eq3}), we have:
\begin{multline*}
F_{q_1,q_2}(\omega_1,\omega_2) \\ =  \int_{0}^{\infty}e^{-\Delta_1+\Delta_2+\Delta_3}\left( \int_{\frac{1-\omega_1}{1+\omega_1}}^{\infty} d\Delta_1 \int_{\frac{1-\omega_2}{1+\omega_2}}^{\infty} d\Delta_3 \right) d\Delta_2
\end{multline*}
Similarly, we get the distribution function of $F_{q_2,q_3}(\omega_2,\omega_3)$. With the joint distributions of $q_1,q_2$ and $q_2,q_3$, we can deduce the joint density of $q^j_1, q^j_2,q^j_3$ as:
\begin{multline}\label{Eq6}
f\left( q^j_1, q^j_2,q^j_3\right)  \\  =  \frac{128 \left( 1-q^j_1\right) \left( 1+q^j_2\right)\left( 1-q^j_2\right) \left( 1+q^j_3\right)} {\left(3-q^j_1+q^j_2+q^j_1 q^j_2\right)^3 \left(3-q^j_2+q^j_3+q^j_2 q^j_3\right)^3}
\end{multline}
Using Eqs.~(\ref{Eq4}\&\ref{Eq6}), we have the joint distribution of $r^{j+1}_1,r^{j+1}_2,r^{j+1}_3,q^j_1,q^j_2,q^j_3$. Further assuming the marginal distribution of $r^{j+1}_1, r^{j+1}_2,r^{j+1}_3$ to be normally distributed {(see Proposition 1)}, the covariance matrix can be numerically determined as follows:
\begin{multline*}
\sum=\int_{-1}^{1} dq^j_1 \int_{-1}^{1} dq^j_2 \int_{-1}^{1} dq^j_3 \sum\left( q^j_1, q^j_2,q^j_3\right) p\left( q^j_1, q^j_2,q^j_3\right) \\
 \approx \left( 
\begin{matrix}
0.42 &  0.25  & 0.058\\
0.25 &  0.42  & 0.25\\
0.058 &  0.25 & 0.42 \\
\end{matrix}\right)
\end{multline*}
Thus, the marginal distribution $r^{j+1}_1, r^{j+1}_2,r^{j+1}_3$ can be written as follows:
\begin{equation*}
f\left( r^{j+1}_1, r^{j+1}_2,r^{j+1}_3\right)  \approx \frac{1}{\sqrt{8\pi^3\text{Det}(\sum)}}\exp\left(-\frac{1}{2}\textit{\textbf{r}}^{\textit{T}} \left(\sum\right)^{-1}\textit{\textbf{r}}\right)
\end{equation*}
Once we have the distribution functions, we calculate the probability of retaining an extremum in level $j+1$ as,
\begin{multline}\label{Eq7}
\int_{-\infty}^{\infty}dr^{j+1}_3\int_{r^{j+1}_3}^{\infty}r^{j+1}_2 \int_{-\infty}^{r^{j+1}_2}p\left( r^{j+1}_1, r^{j+1}_2,r^{j+1}_3\right) dr^{j+1}_1 \\ \approx 0.24 
\end{multline}
Upon generalizing Eq.~(\ref{Eq7}), we get the probability of retaining an extremum over $\eta$ subsequent levels as $0.24^{\eta}$.   \QED }
\end{proof}

\section{Proof of Proposition 1}
\label{appendix:B}		

\begin{proposition}
Let $ r^{j+1}_k $ be the extremum in the rotation component $R^{j+1}(t)$ at any decomposition level $j+1$ of $l_k$. Then the distribution function of $r^{j+1}_{k}$ is given by the convolution of three independent random variables $K_1, K_2$, and $\Gamma$ such that:
\begin{equation*}
F_{r^{j+1}_{k}}(r)= \hskip -2em \underset{\begin{subarray}{c}\lbrace(\kappa_1,\kappa_2,\gamma) \in\mathbb{R}^2 \times[0,2\nu\sigma];\\ \kappa_1+\kappa_2+\gamma\leq r\rbrace\end{subarray}}{\int\int\int  } \hskip -2em F_{K}(d\kappa) F_{K}(d{\kappa})F_{\Gamma}(d\gamma)
\end{equation*}
where $K_i,i=1,2$ are identically distributed and can be represented as a sum of independently distributed normal random variables $l^j_k\sim \mathcal{N}(0,\sigma^2)$ and $\Theta$ as: 
\begin{equation*}
F_K(l,\theta) = \hskip -2em \underset{\{(l,\theta)\in\mathbb{R}^2:l+\theta\leq\kappa\}}{\int\int} \hskip -2em G_{\Theta}(d\theta) G_{l^j_k}(dl)
\end{equation*} 
where the distribution function of $\Theta$ is given as:
$$G_{\Theta}(\theta)=\int_{-\infty}^{\theta}\left(\int_{\infty}^{\infty}f_{U,l^j_k}\left(l,\frac{\omega}{l}\frac{1}{|l|}dl \right)  d\omega\right)$$
with $U\sim uniform(0,2)$ and $l^j_k\sim \mathcal{N}(0,\sigma^2)$. $\Gamma$ follows a mixture distribution such that:
$$F_{\Gamma}(\gamma)=\int_{0}^{\gamma}\frac{1}{2\nu\sigma}d\omega 1_{k={k^*}\pm1}+c 1_{k=k^*}$$
where $c>0$.
\end{proposition}

\begin{proof}{}\label{one}
\normalfont{Using the compact notation introduced in \cite{restrepo2014defining}, we can represent the extrema vector, ${\textit{\textbf{l}}}^{j+1}=[ l^{j+1}_k]_{k=1,2,\ldots,N}$  in the baseline component at level $j+1$ as: 
\begin{equation}\label{Eq8}
{\textit{\textbf{l}}}^{j+1}=\mathcal{T}\left(\tilde{\textit{\textbf{l}}}^{j+1}\right)
\end{equation} where $\mathcal{T}$ is an extrema extracting operator such that $\tilde{\textit{\textbf{l}}}^{j+1}=\left(I+M^j\right){\textit{\textbf{l}}}^j$ and $M^j$ is the tri-diagonal matrix as follows:
$$ M^j=\left( 
\begin{matrix}
2 &  2  & 0  & \ldots & 0\\
1-q^j_2 &  2  & 1+q^j_2  & \ldots & 0\\
0 &  1-q^j_3  & 2  & \ddots & \vdots\\
0 &  0  & 0  & \ldots & 2\\
\end{matrix}\right);\textit{\textbf{l}}^j=\left( \begin{matrix}
l^j_1 \\ l^j_2 \\ \vdots \\ {l}^j_N  
\end{matrix}\right)
$$
Thus, Eq.~(\ref{Eq8}) can be rewritten as:
\begin{multline}\label{Eq9}
{\textit{\textbf{l}}}^{j+1}=\mathcal{T}\left[(I+M^j){\textit{\textbf{l}}}^j+\nu\sigma(I+M^j)e_{k^*} \right]  \\ = \mathcal{T}\bigg[\left(I+M^j\right){\textit{\textbf{l}}}^j+ \frac{1}{4}\left[ \begin{matrix}
\b{0} \\ ( 1+q^j_{k^*-1})\nu\sigma \\2\nu\sigma \\( 1-q^j_{k^*+1})\nu\sigma\\ \b{0}
\end{matrix} \right]\bigg] 
\end{multline}
where $e_{k^*}=\left[\begin{matrix}
\b{0} & 1_{\{k=k^*\}} & \b{0}\end{matrix}\right]^{T}$. Consequently the ``$\nu\sigma$" containing terms in Eq.~(\ref{Eq9}) are:
{\begin{equation*}
\hskip -1em \left[ \begin{matrix}
\tilde{l}^{j+1}_{k^*-1} \\ 
\tilde{l}^{j+1}_{k^*} \\ 
\tilde{l}^{j+1}_{k^*+1} 
\end{matrix}\right] 
=\frac{1}{4} 
\begin{bmatrix}
q^{j-}_{k^*-1} l^j_{k^*-2} +2l^j_{k^*-1} +q^{j+}_{k^*-1}l^j_{k^*}
{} +q^{j+}_{k^*-1} \nu\sigma \\ \\
q^{j-}_{k^*} l^j_{k^*-1} +2l^j_{k^*} + q^{j+}_{k^*} l^j_{k^*+1} 
{}+2\nu\sigma \\ \\
q^{j-}_{k^*+1} l^j_{k^*} +2l^j_{k^*+1} + q^{j+}_{k^*+1}l^j_{k^*+2} 
{} +q^{j-}_{k^*+1} \nu\sigma\\
\end{bmatrix} 
\end{equation*}}where, $q^{j-}_k=1-q^j_{k}$ and $q^{j+}_k=1+q^j_{k}$. Notice that without the operator $\mathcal{T}$ in Eq.~(\ref{Eq9}), terms on the LHS may not be guaranteed to be extrema (see Property S1 in Appendix F).
The corresponding points in the rotation components are given as follows:
\begin{multline}\label{Eq10}
 \tilde{r}^{j+1}_k=\begin{cases}
\dfrac{1}{4} \left( 2l^j_{k} -q_k^{j-} l^j_{k-1} - q_k^{j+} l^j_{k+1}\right),\\ 
\qquad \qquad \qquad  \qquad k^*-1>k>k^*+1 \\
\dfrac{1}{4}\left( 2l^j_{k} -q_k^{j-} l^j_{k-1} - q_k^{j+} l^j_{k+1}\right) +f_{k},  \\ \qquad \qquad \qquad  \qquad k^*-1\leq k\leq k^*+1
\end{cases}
\end{multline}
Again, $\{\tilde{r}^{j+1}_k\}_{k=1,2,\ldots,N}$ represent only the corresponding values of $r^j_k$ in level $j+1$ and not the extrema points. The term $f_{k}$ in Eq.~(\ref{Eq11}) represents the the effect of scaled Kronecker delta $\nu\sigma\delta_{k^*}$ (at $k^*$ in level $j$) at locations $k^*-1,k^*$ and $k^*+1$ in level $j+1$ such that: 
\begin{equation}\label{Eq11}
f_{k}=\begin{cases}
\nu\sigma/2 & k = {k^*}\\
-q^{j\mp}_{k^*\mp 1}\nu\sigma/4 & k = {k^*\mp 1} \\
0& \text{o.w.}
\end{cases}
\end{equation}
{\color{black}From Lemma 2, we notice that $q^j_{k^*}\sim$~uniform(-1,1). Therefore, $1+q^j_{k^*}$ and $1-q^j_{k^*}$ follows uniform$(0,2)$ distribution.} Let us define, $\Theta:=\left( 1\pm q^j_{k}\right) l^j_{k+1}$ with distribution function $G_{\Theta}$ where $l^j_{k^*+1}\sim N(0,\sigma^2)$. Therefore, $G_{\Theta}$ is the product distribution given as follows:
\begin{equation*}
G_{\Theta}(\theta)=\int_{-\infty}^{\theta}\left(\int_{\infty}^{\infty}f_{U,l^j_k}\left(l,\frac{\omega}{l}\frac{1}{|l|}dl \right)  d\omega\right)
\end{equation*}
Next, we define $K_1:=l^j_k - ( 1- q^j_{k+1}) l^j_{k-1}$, which is the sum of normal random variable, $l^j_k\sim G_{l^j_k}(l)$ and $\Theta\sim G_\Theta(\theta)$. Similarly, we define, $K_2:=l^j_k-(1+ q^j_{k-1}) l^j_{k-1}$ such that: 
\begin{equation*}
F_K(l,\theta) =\int\int_ {\{(l,\theta)\in\mathbb{R}^2:l+\theta\leq\kappa\}} G_{\Theta}(d\theta) G_{l^j_k}(dl)
\end{equation*} 

Now, from Eq.~(\ref{Eq10}), we have $r^{j+1}_k\left(=\mathcal{T}[\tilde{r}^{j+1}_k] \right) $ as the sum of $K_1, K_2$ and $\Gamma \left( =f_{k}\right) $. Also, from the definition of $\Gamma$ in Eq.~(\ref{Eq11}), we have $K_1,K_2$ and $\Gamma$ are independently distributed where $\Gamma$ is a mixture distribution given as 
\begin{equation*}
F_{\Gamma}(\gamma)=\int_{0}^{\gamma}\frac{1}{4\nu\sigma}d\omega 1_{{t=\tau^j_{k^*\pm1}}}+\frac{\nu\sigma}{2} 1_{{t=\tau^j_{k^*}}}
\end{equation*}
Combining the above results, we have 

\hfill\begin{equation*}
\displaystyle F_{r^{j+1}_{k}}(r)= \hskip -2em \underset{\begin{subarray}{c}\lbrace(\kappa_1,\kappa_2,\gamma) \in\mathbb{R}^2 \times[0,2\nu\sigma];\\ \kappa_1+\kappa_2+\gamma\leq r\rbrace\end{subarray}}{\int\int\int  } \hskip -2em F_{K}(d\kappa) F_{K}(d{\kappa})F_{\Gamma}(d\gamma) 
\end{equation*}  } \vskip -4.7em \QED \vskip 3.7em 
\end{proof}

\section{Proof of Corollary 2} \label{appendix:C}
\begin{lemma}
The exact variance of the product of two independent random variables $X$ and $Y$ is given as, $$\text{Var}(XY) = (\mathbb{E}(X))^2  \text{Var}(Y) + (\mathbb{E}(Y))^2  \text{Var}(X) + \text{Var}(X)\text{Var}(Y)$$
\end{lemma}
\begin{proof}
See [2] for proof. 
\end{proof}
\begin{corollary}
Using Gaussian approximations to the distribution function of $r^{j+1}_{k}$, $P_e(\nu)$ can be deduced in closed form as:  
\begin{equation*}
\hat{P}_e(\nu)=\left[ 1-P\left(\mathcal{Z}\leq-\frac{\nu}{\sqrt{2}} \right)\right]^2  
\end{equation*}
where $\mathcal{Z}\sim \mathcal{N}(0,1)$.
\end{corollary}
\begin{proof}{}\label{one}

\normalfont{Here, we are interested in the Gaussian approximation to the distribution of $ ({r}^{j+1}_k- {r}^{j+1}_{k+1}) $ for $j>1$ where $r^{j+1}_{k}$ is represented as follows:
\begin{multline}\label{Eq12}
{r}^{j+1}_k=	\begin{cases}
\dfrac{1}{4}\left(2l^j_k-q^{j-}_kl^j_{k-1}-q^{j+}_kl^j_{k+1}\right),  \\ \qquad \qquad \qquad  \qquad k^*-1>k>k^*+1\\
\dfrac{1}{4}\left(2l^j_k-q^{j-}_kl^j_{k-1}-q^{j+}_kl^j_{k+1}\right)+{f}_{k},  \\ \qquad \qquad \qquad  \qquad k^*-1\leq k\leq k^*+1
\end{cases}
\end{multline}
where $q^{j\pm}_k\sim\text{uniform}(0,2)$ and $l^j_k$ is defined as: 
$$l^j_k = \frac{1}{4}\left((l^{j-1}_{k-1}+l^{j-1}_{k+1}) + 2l^{j-1}_k+q^{j-1}_k(l^{j-1}_{k+1}-l^{j-1}_{k-1})\right)$$

For the given signal $x_k(\equiv l^0_k)$, we have $\mathbb{E}[l^{0}_k]=0$ and Var$(l^{0}_k)=\sigma^2$. For simplicity, we assume $\sigma =1$. As the extrema points evolve via a random decimation process, we can assume that the inter-extremal separations and the extrema values are independent of each other. From \cite{restrepo2014defining}, we note that the density function of $q^{j-1}_k(l^{j-1}_{k-1}-l^{j-1}_{k+1})$ is approximately normal with ${\mathbb{E}}[q^{j-1}_k(l^{j-1}_{k-1}-l^{j-1}_{k+1})]=0$ and Var$(q^{j-1}_k(l^{j-1}_{k-1}-l^{j-1}_{k+1}))=\text{Var}(q^{j-1}_k)\text{Var}(l^{j-1}_{k-1}-l^{j-1}_{k+1})=2/3$ (see Lemma~3), such that we have: 

$$l^j_k\sim \mathcal{N}\left(0,\frac{5}{12}\right)$$

Next, we determine the distribution of ${r}^{j+1}_k$. For $k=k^*$, the expected value of $r^{j+1}_k$ is $\mathbb{E}[{r}^{j+1}_k]= \nu/2$ (see Eq.~\eqref{Eq11}) and the variance term using Lemma 3 is given as follows: 
\begin{align*}
\text{Var}&({r}^{j+1}_k) = \frac{1}{16}\text{Var}\left(2l^j_k-q^{j-}_kl^j_{k-1}-q^{j+}_kl^j_{k+1}\right)\\
&= \frac{1}{16}\bigg( \frac{20}{3}\sigma^2 - 4\mathbb{E}[q^{j-}_k]\text{Cov}(l^j_k,l^j_{k-1}) \\ & + 2 \mathbb{E}[q^{j-}_k]^2\text{Cov}(l^j_{k-1}, l^j_{k+1}) - 4\mathbb[q^{j+}_k]\text{Cov}(l^j_k,l^j_{k+1})\bigg)\\
&= \frac{1}{16}\left(\frac{20}{3}\sigma^2 -2 + \frac{1}{12}\right) = \frac{1}{16}\left(\frac{25}{9}-\frac{23}{12}\right)\approx \frac{1}{19}
\end{align*}
In summary, we have:
\begin{equation*}
{r}^{j+1}_k\sim	\begin{cases}
\mathcal{N}\left(\dfrac{\nu}{2},\dfrac{1}{19}\right) & k=k^*\\
Y+Z & k = \{k^*-1,k^*+1\}\\
\mathcal{N}\left(0,\dfrac{1}{19}\right) & k^*-1>k>k^*+1\\
\end{cases}
\end{equation*}where $Y\sim \mathcal{N}(0,1/19)$ and $Z\sim\text{uniform}(-\nu/4,0)$. Let ${r}^{j+1}_{k^*}= X$, such that the distribution of $r^{j+1}_{k^*}- {r}^{j+1}_{k^*+1}$ is equivalent to $X-(Y+Z)$.

We now determine the distribution of $X-Y$. Here, $\mathbb{E}[X-Y]=\nu/4$ and \text{Var}(X-Y) given as:
\begin{align*}
& \text{Var}(X)+\text{Var}(Y)-2\text{Cov}(X,Y)\\
&=2\times\frac{1}{19} + \frac{1}{16}\mathbb{E}[4l^j_kl^j_{k-1}-2l^j_kl^j_{k-2}-2l^j_kl^j_{k} \\&-  2l^j_{k-1}l^j_{k-1} + l^j_{k-1}l^j_{k-2}  + l^j_{k-1}l^j_{k} - 2l^j_{k-1}l^j_{k+1} + l^j_{k+1}l^j_{k}] \\
&=2\times\frac{1}{19} + \frac{1}{16}\left(\frac{7}{4}-\frac{4}{16}-\frac{20}{12}\right) = 2\times\frac{1}{19}-\frac{1}{96}\approx0.01
\end{align*}

\noindent Now, the distribution of $X-Y+Z$ can be expressed as the sum of: $$X-Y\sim\mathcal{N}\left(\frac{\nu}{4},\frac{1}{8}\right)~ \& ~Z\sim\text{uniform}\left(0,\frac{\nu}{4}\right)$$ Let $X-Y=X'$. We can now determine the probability $P(X'+Z>0)$ as the following conditional probability:
\begin{align}\label{Eq14}
P&(X'+Z>0)  =P(X'+Z>0|X'>0)P(X'>0) \nonumber\\& + P(X'+Z>0|X'<0)P(X'<0)\nonumber\\
& =P(X'>0)+P(X'+Z>0|X'<0)P(X'<0)
\end{align} 
Using the convolution of uniform and normal random variables, we re-write Eq.~(\ref{Eq14}) as: 
\begin{multline}\label{Eq15}
	P(X'+Z>0|X'<0)P(X'<0)  \\= P(X'<0)-P(X'+Z<0)
\end{multline}
We estimate the probability, $P(X'+Z<0) -P(X'<0)$ as shown in Eq.~\eqref{eq:floatingeq} on top of page 5.
\addtocounter{equation}{1}%
\setcounter{storeeqcounter}%
{\value{equation}}%

\begin{figure*}[!t]
	\setcounter{tempeqcounter}{\value{equation}} 
	\begin{IEEEeqnarray}{rCl}
		\setcounter{equation}{\value{storeeqcounter}} 
		P(X'<0)- P(X'+Z<0) = \frac{\nu}{2\sqrt{\pi}}\int_{-\infty}^{0} \left(\exp\left(-\frac{\nu^2}{4}(x'-1)^2\right)- \int_{0}^{1}\exp\left(-\frac{\nu^2}{4}(z-x'+1)^2\right)dz\right)dx' \nonumber \\	
		=\frac{\nu}{2\sqrt{\pi}}\int_{-\infty}^{0}  \exp\left(-\frac{\nu^2}{4}(x'-1)^2\right)- {\frac{\sqrt{\pi}}{\nu}\left(\mathop{\mathrm{erf}}\nolimits\!\left(\frac{\nu \left(x' - 1\right)}{2}\right) - \mathop{\mathrm{erf}}\nolimits\!\left(\frac{\nu\left(x' - 2\right)}{2}\right)\right)}dx'  
		\label{eq:floatingeq}
	\end{IEEEeqnarray}
	\setcounter{equation}{\value{tempeqcounter}} 
	\hrulefill
\end{figure*}

For the values of $1<\nu<4$, we note that the function $\int_{0}^{1}\exp\left(-\frac{\nu^2}{4}(z-x'+1)^2\right)dz$ is approximately normal and centered at $1.5$ with variance approximately equal to Var$(X')$. Since the first term is centered at 1, therefore the difference is non-zero. However, the magnitude of difference supported on negative $x'$ axis is $\leq 0.02$ and decays exponentially fast such that the magnitude of the difference is approximately 0 for $z\leq -0.5$. Hence the integral of the difference remains sufficiently close to zero and thus can be ignored. Therefore, $$P(X'+Z>0|X'<0)P(X'<0)\approx0$$  Hence, we can approximate the probability, $P_e(\nu)$ as 
\begin{equation*}
\hat{P}_e(\nu)=P(X'>0) = \left(1-P\left(\mathcal{Z}\leq-\frac{\nu}{\sqrt{2}}\right)\right)^2
\end{equation*} \vskip -2.5em \QED \vskip 2.5em}
\end{proof}

\section{Proof of Corollary 3 }
\label{appendix:D}
\begin{corollary}
Using the Gaussian approximation to the distribution function of $r^{j+1}_k$ (see Appendix C), $P_s(\nu)$ can be approximated as, 
\begin{equation*}
\hat{P_s}(\nu)=1-P\left( \mathcal{Z} \leq 3-\nu \sqrt{\frac{19}{16}}\right)
\end{equation*}
\end{corollary} 

\begin{proof}{}\label{one}
\normalfont	{Since, we have $\hat{r}^{j+1}_{k^*}\sim \mathcal{N}\left(\frac{\nu}{4},\frac{1}{19} \right)$, therefore, $\sigma^{j+1}=\sqrt{{1}/{19}}\sigma^j $ and we have, 
\begin{eqnarray*}
\hat{P_s}(\nu) &=& P(\hat{r}^{j+1}_{k^*}>3\sigma^{j+1}|\nu\geq 3) \\
&= & 1-P\left(\mathcal{Z}\leq\frac{\left(3\sqrt{\frac{1}{19}}-\frac{\nu}{4}\right)}{\sqrt{\frac{1}{19}}}\right)\\
&=& 1-P\left(\mathcal{Z}\leq 3-\nu \sqrt{\frac{19}{16}}\right)   
\end{eqnarray*} \vskip -2.5em \QED \vskip 2.5em}
\end{proof}

\section{Proof of Corollary 4}		
\label{appendix:E}

\begin{corollary} The probability $P(\sigma_0, \sigma_a, n_0, n_a)$ that the rotation component $\mathbf{r}^j$ captures the variance shift is given as:
$$P({\sigma_0},{\sigma_a}, n_0, n_a) = P\left(\frac{\textit{S}(r^j_{k\leq k^*})}{\textit{S}(r^j_{k> k^*})} < 1\right)= \mathcal{B}_{\varkappa}\left(\frac{n_0}{2},\frac{n_a}{2}\right)$$ where $S(.)$ denotes the sample variance, $\mathcal{B}$ is the regularized incomplete beta function evaluated at $\varkappa = {n_0\sigma_a^2}/({n_0\sigma_a^2 + n_a\sigma_0^2})$ with $n_0+1$ and $n_a+1$ being the length of time series in the in-control and out of control region, respectively.  
\end{corollary}
\begin{proof} \normalfont Let $S_0 = {\text{{Var}}}(r^j_{k\leq k^*})$ and $S_a ={\text{{Var}}}(r^j_{k> k^*})$. Using, the distribution of sample variance, we have 
$$S_0\sim \frac{\sigma_0^2}{n_0} \chi^2(n_0)$$
$$S_a\sim \frac{\sigma_a^2}{n_a} \chi^2(n_a)$$ Ratio of $\chi^2$ distributed radom variables with degrees of freedom $n_0$ and $n_a$ follows a $F$- distribution with degrees of freedom $n_0$ and $n_a$ given as, 
$$ S_{0,a}\equiv\frac{S_0}{S_a}\sim  \frac{\sigma^2_0}{\sigma^2_a}F\left(\frac{n_0}{2},\frac{n_a}{2}\right)$$ Therefore, 

\begin{align*}
P(\sigma_0, \sigma_a, n_0, n_a) & = P\left(\frac{S_0}{S_a} <1\right) \equiv P\left(S_{0,a} <  \frac{\sigma^2_a}{\sigma^2_0}\right) \\&= \mathcal{B}\left({\frac{n_0\sigma_a^2}{n_0\sigma_a^2 + n_a\sigma_0^2}};{n_0},{n_a}\right)
\end{align*} \vskip -2.5em \QED \vskip 2.5em
\end{proof}
\section{Proof of Proposition 2}
\label{appendix:F}
\begin{proposition}
	\textit{ The ratio of expected value of phase synchronization when there is a singularity at $k=k^*$ to the case when there is no singularity at $k=k^*$, i.e.,  
		\begin{eqnarray}\label{eq:15}
		\xi=\frac{\mathbb{E}\left[\Phi^{j,j+1}_{k} \big|r^{j}_{k^*}\geq 3\sigma^j\right] }{\mathbb{E}\left[ \Phi^{j,j+1}_{k}\big|r^j_{k^*}<3\sigma^j\right] }
		\end{eqnarray}}
	is lower bounded as: $$\xi \geq P_s(\nu|\nu>3\sigma^j)\lim_{h\to 0}({P_e(h)})^{-1}\approx4P_s(\nu|\nu>3\sigma^j)$$
\end{proposition}

Before we present the proof, we first present the following property of halfwaves.
\begin{customppty}{S1}\label{five}
	Each Extremum, $\{r^j_k\}_{k=1,2,\ldots,N}$ in any rotation component $R^j(t)$ can evolve in $R^{j+1}(t)$ via one of the three transitions: (a) no change in extremum orientation (call this as trans-critical) as shown in Fig.~(\ref{fig:prop3}(a)), vanishing of the extremum (saddle-node) as shown in Fig.~(\ref{fig:prop3}(b)) or flipping of the extremum (pitch-fork) as shown in Fig.~(\ref{fig:prop3}(c)). The probabilities of each of these transitions calculated using the Gaussian approximation to the distribution function of $r^{j+1}_k$ are 0.25, 0.5 and 0.25, respectively (also see \cite{restrepo2014defining}). 
\end{customppty}
Here, the probability of trans-critical transition is equivalent to the probability of retaining an extremum in level $j$ as an extremum in level $j+1$ and is equal to 0.25 (See Corollary 1). To determine $P(\eta_1)$, we look at the transitions as shown in Figs.~\ref{fig:prop3}($f^{1},...,f^6$).  $P(\eta_1)$ can be calculated as the sum of $ P(r^{j+1}_{k^*}<r^{j+1}_{k^*-1})P(r^{j+1}_{k^*}>r^{j+1}_{k^*+1})$ (for events $f^2, f^3~\&~f^6$) and $ P(r^{j+1}_{k^*}>r^{j+1}_{k^*-1})P(r^{j+1}_{k^*}<r^{j+1}_{k^*+1})$  (for events $f^1, f^4~\&~f^5$) and is equal to 0.5. Finally, the probability of pitchfork transition, $P(\eta_2)$ can be determined by calculating the probability of event $r^{j+1}_{k^*}-r^{j+1}_{k^*-1}<0$ and $r^{j+1}_{k^*}-r^{j+1}_{k^*+1}<0$ simultaneously. From the Gaussian approximation to the distribution function of $r^{j+1}_k$, we have $P(\eta_2) = P(r^{j+1}_{k^*}<r^{j+1}_{k^*-1})P(r^{j+1}_{k^*}<r^{j+1}_{k^*+1}) =  0.25 $. 
\begin{figure*}[t]
	\includegraphics[width = 0.93\textwidth]{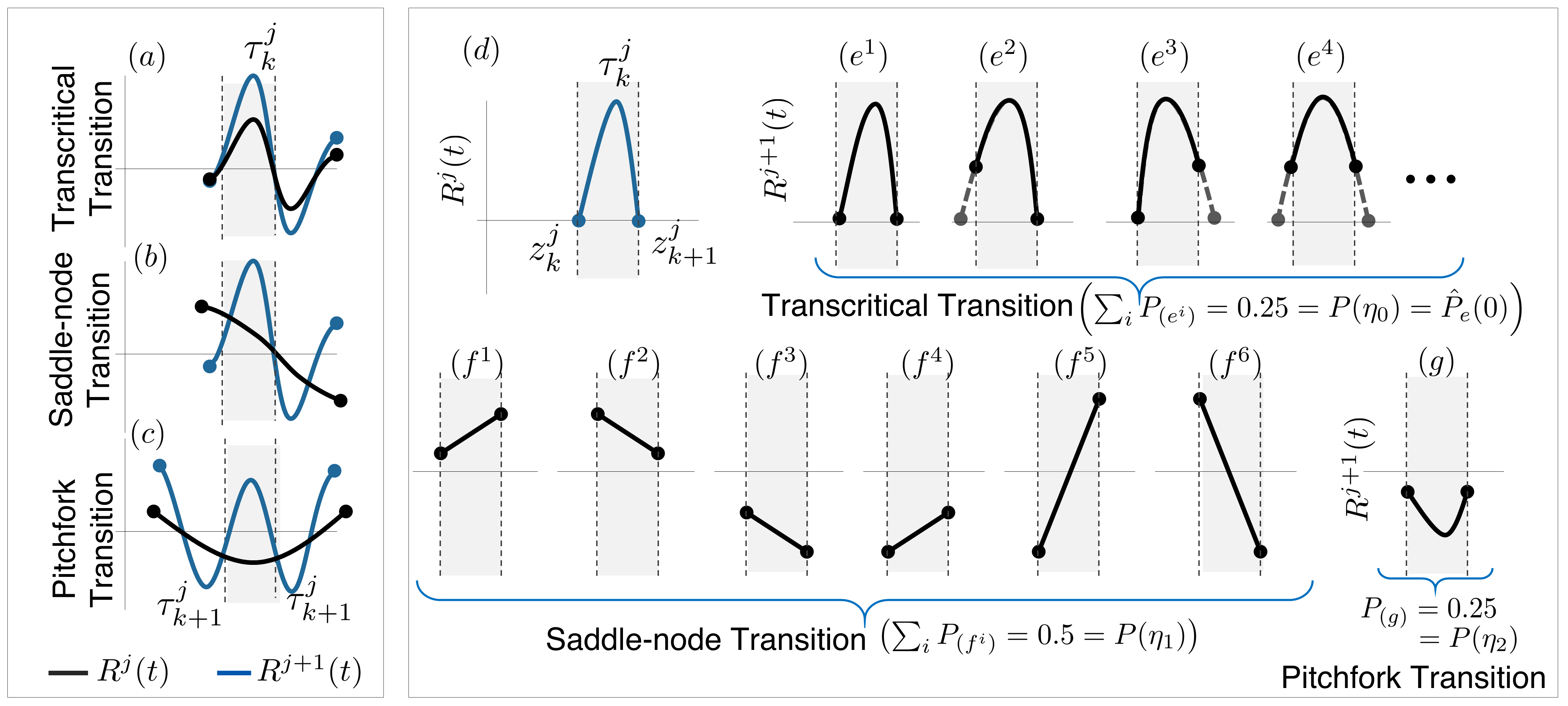}
	\centering
	\caption{(a-c) shows the evolution of halfwave from level $j$ to $j+1$ via transcritical, saddle-node, and pitchfork transition, respectively. (d) A representative halfwave, $\hslash^j_k(t)$ in level $j$ with characteristic extremum at $\tau^j_{k}$. (e$^{1}$, e$^2$, ...) shows a few cases of evolution of $\hslash^j_k(t)$ via transcritical transition, (f) shows the orientation of $\hslash^j_k(t)$ in level $j+1$ resulting due to saddle-node transition and (g) shows the evolution of $\hslash^j_k(t)$ via pitchfork transition. The probability of each of these cases are determined using the Gaussian approximation to the distribution function of $\hat{r}^{j+1}_k$.}	
	\label{fig:prop3}
\end{figure*} 

We now look at the proof of Proposition 2.

\begin{proof}

\normalfont In order to determine the expected phase synchronization $\Phi^{j,j+1}_k$ between halfwave at levels $j$ and $j+1$, we first identify the fraction of halfwave $\hslash^{j+1}_k(t)$ enclosed within the support $\text{supp}(\hslash^{j}_{k}(t))$. This is represented by the shaded region in Fig.~\ref{fig:prop3}(d). Assuming that the extremum $\tau^j_{k}$ at level $j$ is retained in the next level, then its neighboring extrema may evolve in the next level according to either extrema preserving or extrema vanishing transition (see Property 4). Here, $\text{supp}(\hslash^{j+1}_{k}(t))=(z^{j+1}_{k},z^{j+1}_{k+1}]$ where $z^{j+1}_{k}$ and $z^{j+1}_{k+1}$ are variables and depend on the location of $r^{j+1}_{k{\pm1}}$. Under the given assumptions, the possible cases for the evolution of $\hslash^{j}_k(t)$ in level $j+1$ are as shown in Fig.~\ref{fig:prop3}(e$^{1}$, e$^2$, ...). These are (e$^{1}$) where all the extrema, i.e., $\tau^j_{k}$ and $\tau^j_{k\pm 1}$ are retained; (e$^{2}$) where only the minimum at $\tau^j_{k-1}$ vanishes, hence shifting $z^{j+1}_{k}$ towards left; (e$^{3}$) where minimum at $\tau^j_{k+1}$ vanishes causing $z^{j+1}_{k+1}$ to shift towards right and (e$^{4}$) where the minima on either side of $\tau^j_{k}$ vanishes, support of $\hslash^j_{k}$ on both the direction increases, and so on.

To identify the fraction of halfwave $\hslash^{j+1}_k(t)$ enclosed within the support $\text{supp}(\hslash^{j}_{k})$, consider the halfwave $\hslash^j_k(t)$ which is characterized by the points $\{R^j(z^j_{k}), r^j_{k}, R^j(z^j_{k+1})\}\equiv\{0, r^j_{k}, 0\}$. Similarly, the points $\{R^{j+1}(z^j_{k}), r^{j+1}_{k}, R^{j+1}(z^j_{k+1})\}$ define the corresponding halfwave in the next level, i.e., $\hslash^{j+1}_{k}(t)$ enclosed within $(z^j_{k},z^j_{k+1}]$. Here, $R^{j+1}(z^j_{k})$ and $R^{j+1}(z^j_{k+1})$ are the amplitudes of $R^{j}(z^j_k)$ and $R^{j}(z^j_{k+1})$, i.e., the amplitudes of zero crossings $z^j_k$ and $z^{j }_{k+1}$ in level $j+1$. We use a linear interpolation to determine the values of $R^{j+1}(z^j_{k})$ and $R^{j+1}(z^j_{k+1})$ as follows: $$R^{j+1}(z^j_{k})=r^{j+1}_{k-1}+\frac{R^{j+1}(\tau^j_{k})-R^{j+1}(\tau^j_{k-1})}{\tau^j_{k}-\tau^j_{k-1}}\big(z^j_{k}-\tau^j_{k-1}\big)$$Since phase is invariant of translation, $\hslash^{j+1}_{k}(t)$ can be translated and equivalently represented by the points $\{0, r^{j+1}_{k}-R^{j+1}(z^j_{k}), R^{j+1}(z^j_{k+1})-R^{j+1}(z^j_{k})\}$. 
To determine the expected phase synchronization, it would suffice to determine the inner product between the halfwaves $\hslash^{j}_{k}(t)$ and $\hslash^{j+1}_{k}(t)$ within the support of $\hslash^{j}_{k}(t)$. 

Considering a singularity at $k^*$, the expected level of phase synchronization can be calculated as follows:
\begin{eqnarray}\label{Eq17}
\mathbb{E}\left[\Phi^{j,j+1}_{k^*}|r^j_{k^*} \geq  3\sigma^j\right]=\eta_0 P_s(\nu)
\end{eqnarray}
\noindent where $\eta_0$ is the value of phase synchronization when extremum at $\tau^j_{k^*}$ is preserved and $ P_s(\nu)$ is the probability that the singularity at $k^*$ is retained in level $j+1$. For the case when $\tau^j_{k^*}$ is not a singularity, the probability of observing the transitions as shown in Fig.~\ref{fig:prop3}(e$^{1}$, e$^2$, ...) depends on the probability of retaining the extremum at $\tau^j_{k^*}$ as an extremum in level $j+1$. For all other cases when extremum is not preserved is shown in Fig.~\ref{fig:prop3}(f-g). So, we can write the expected phase synchronization for this case as:
\begin{multline}\label{Eq18}
\mathbb{E}\left[\Phi^{j,j+1}_{k^*}|r^j_{k^*}<3\sigma^j \right]  = \eta_0 P_e(\nu|\nu=0) P(\mathcal{Z}<3\sigma^j)   \\+\eta_1P(\eta_1)+\eta_2P(\eta_2)
\end{multline}   
Here, $\eta_1P(\eta_1)$ and $\eta_2P(\eta_2)$ are the expected value of phase synchronization when the extremum at $\tau^j_{k^*}$ evolves via a saddle-node  (extremum at $\tau^j_{k^*}$ is not retained, Fig.~\ref{fig:prop3}(f)) and pitchfork transition (extremum at $\tau^j_{k^*}$ is flipped, Fig.~\ref{fig:prop3}(g)), respectively. Under pitchfork transition, the halfwave at level $j+1$ is negatively oriented with respect to $\hslash^j_{k^*}(t)$, causing a phase lag of $\pi$. Using the equation for phase synchronization (Eq.~(14) in the main document), we get $\eta_2\approx-1$. 

To determine $\eta_1P(\eta_1)$, we refer to possible orientations of the halfwave  $\hslash^j_{k^*}(t)$ in level $j+1$ resulting due to pitchfork transition, i.e.,  when the extremum at $\tau^j_{k^*}$ is not retained. This is shown in Fig.~\ref{fig:prop3}(f$^{1}$, f$^2$, ...) along with individual probabilities calculated using the distribution function of $\hat{r}^{j+1}_k$. Here, we note that for each orientation, there is an equal likelihood of finding an orthogonal orientation, e.g., Fig.~\ref{fig:prop3}(f$^{1}$) \& \ref{fig:prop3}(f$^{3}$). Using the property of inner product, i.e., for some function $\bar{v}_1$ and $\bar{v}_2$, we have $\left\langle \bar{v}_1,-\bar{v}_2\right\rangle= -\left\langle \bar{v}_1,\bar{v}_2\right\rangle$. This implies, $\eta_1P(\eta_1)= 0$. Therefore, the ratio of expected level of phase synchronization between the halfwaves, $\hslash^j_{k^*}(t)$ at level $j$ and $j+1$ when $\tau^j_{k^*}$ is a singularity (Eq.~(19)) to when it is not a singularity (Eq.~(20)) is given as: 
\begin{eqnarray*}
\xi = \frac{P_s(\nu|\nu\geq 3\sigma^j)\eta_0}{P_e(\nu|\nu=0)	P(\mathcal{Z}<3\sigma^j) \eta_0 + \eta_1P(\eta_1)+\eta_2P(\eta_2)}
\end{eqnarray*} 
\noindent Since $\eta_1P(\eta_1)=0$ and $\eta_2P(\eta_2)<0$ we have, $\eta_1P(\eta_1)+ \eta_2P(\eta_2)<0$ and hence, $\xi \geq  4 P_s(\nu|\nu>3\sigma^j)  $. \QED 

\end{proof}

\section{multiple change point detection}
\label{appendix:G}

In this section, we first present the pseudo code for performing the component selection via mutual agreement concept. This is presented in Algorithm 1.
 
\begin{algorithm}
	\SetKwInOut{Input}{Input}
	\SetKwInOut{Output}{Output}
	
	\underline{function mutualAgreement}$(G=(V,E))$\;
	\Input{$V=\{x(t), R^j(t)\},j = 1,\ldots,J-1;E=[e_{ij}]$}   		
	\Output{Clusters $\mathcal{G}_1, \mathcal{G}_2,\ldots$}
	Estimate $\vartheta_p$ such that $P(E>\vartheta_p)\approx 0.1$\;
	Update $E$ as $E[E<\vartheta_p ]\leftarrow 0$ {\color{darkgray}\%\textit{Eq.~(16) main text}}\;
	$k\leftarrow 1$\;
	\For{every node $R^j(t)\in V \backslash x(t)$}{ 
		$\mathcal{G}_k\leftarrow$DFS($R^j(t)$) {\color{darkgray}\%\textit{use depth-first search to identify components connected to $R^j(t)$}}\;
		$V\leftarrow \{V\backslash \mathcal{G}_k,x(t)$\}\;
		$k\leftarrow k+1$\;
	}
	\caption{Mutual agreement}
\end{algorithm}

\begin{algorithm}
	\SetKwInOut{Input}{Input}
	\SetKwInOut{Output}{Output}
	\underline{function DFS}($R^j(t)$)\;
	$\mathcal{G}_k \leftarrow \{\mathcal{G}_k, R^j(t)\}$ {\color{darkgray}\%\textit{mark node $R^j(t)$ visited}}\;
	\For{every unvisited node $R^i(t)$ for which $e_{ij}>0$}{
		{$\mathcal{G}_k\leftarrow$DFS$(R^i(t))$}}
	\If{$x(t)$ is adjacent to $R^j(t)$}{$\mathcal{G}_k \leftarrow \{\mathcal{G}_k, x(t)\}$ {\color{darkgray}\%\textit{mark node $x(t)$ visited}}}
	\caption{Depth-first search (DFS)}
\end{algorithm}

\begin{figure}[t]
	\includegraphics[width = 0.4\textwidth]{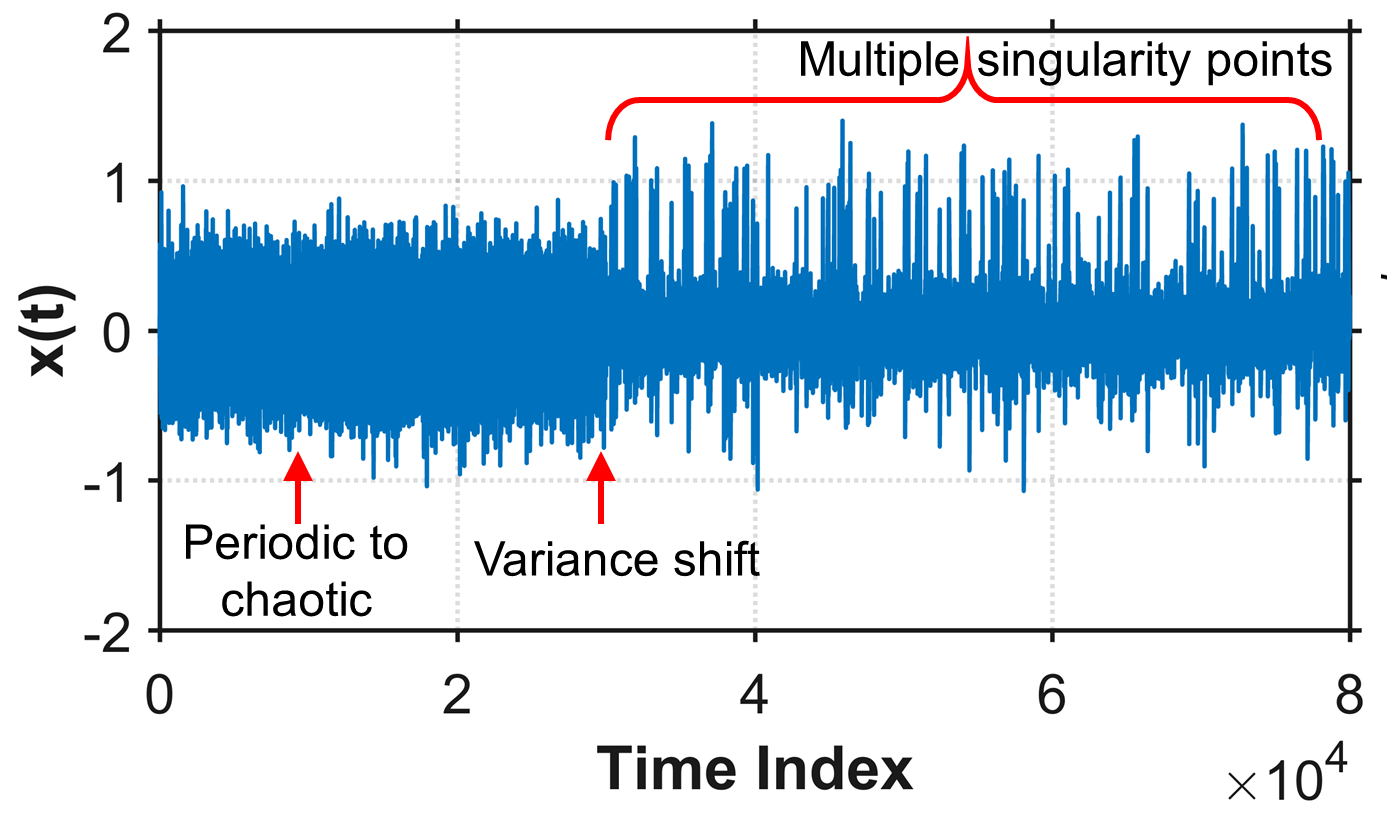}
	\centering
	\caption{Time series consisting of multiple change points including a dynamic pattern change at $t=10000$ t.u. followed by variance shift at $t=30000$ t.u. Subsequent to this point, the multiple singularity points are interspersed.}
	\label{figure:2}	
\end{figure}
\noindent \textbf{Multiple change point detection: }To demonstrate the performance of our method when multiple types of changes are interspersed, we create an artificial data by combining the Logistic map (Section 4.3 in the main text) and the neocortical signal (Section 4.6 in the main text) as shown in Eq.~(\ref{Eq20})
\begin{equation}
\begin{aligned}
z(t) &= y(t) + N(0,\sigma^2) \\     y(t+1) &=\mu y(t)(1-y(t));\mu>0,t\in\mathbb{Z}^+ 
\end{aligned}
\label{Eq19}
\end{equation}

\begin{equation}
x(t) = \begin{cases}
z(t), \mu = \begin{cases}
 3.4  & t\leq 10000 \\ 3.7 & 10000<t<30000
\end{cases} \\ 
z(t), \text{Var}(z(t)) = \begin{cases}
 0.05  & t\leq 30000 \\ 0.025 & t>30000
\end{cases} 
\end{cases}
\label{Eq20}
\end{equation}

As stated, we first generated a 30000 data points long time-series, $x(t)$ from the logistic map for $t<30000$ t.u. as shown in Eq.~(\ref{Eq19}). Here, the first change point is introduced at $t=10000$ t.u. in the dynamic behavior of logistic map from periodic to chaotic by changing the value of $\mu$ from 3.4 to 3.7. The signal to noise ratio is fixed to 10 dB. For the remaining time, i.e, $30000<t<80000$ t.u., we consider the neocortical signal (Section 4.6, main manuscript) consisting of multiple consecutive singularities. This is shown in Fig.~\ref{figure:2}. Here, $t=30000$ t.u is the second change point, followed by multiple change points, each connoting a short-lived change (i.e., a neocortical spike).

\begin{figure}[h]
	\includegraphics[width = 0.5\textwidth]{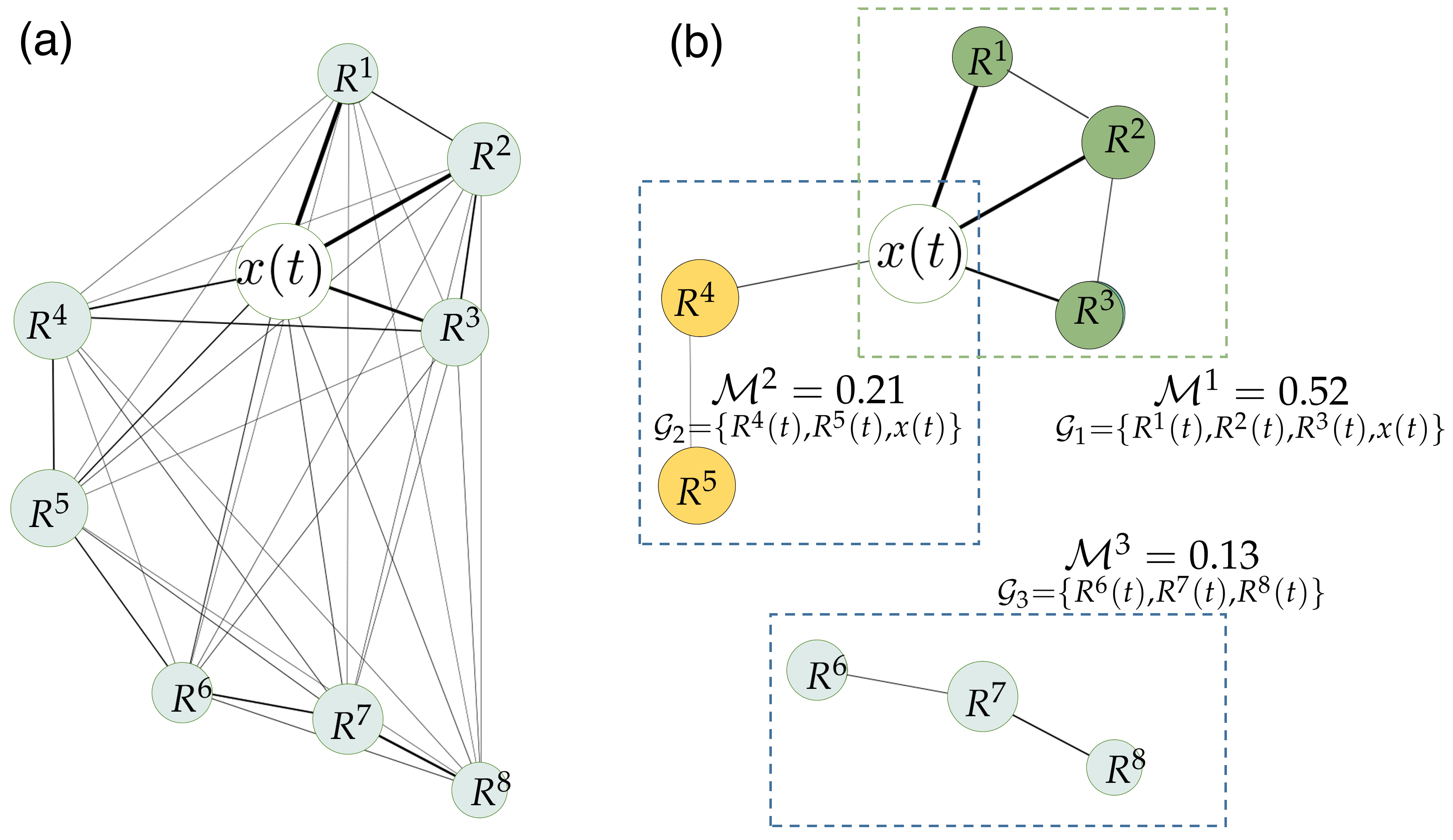}
	\centering
	\caption{(a) Graph representation showing the association between the elements of $G$. (b) Clusters of rotation components obtained after removing the spurious connections as determined by the Pareto threshold $\vartheta_p$.}
	\label{figure:3}	
\end{figure}

\begin{figure}[h]
	\includegraphics[width = 0.40\textwidth]{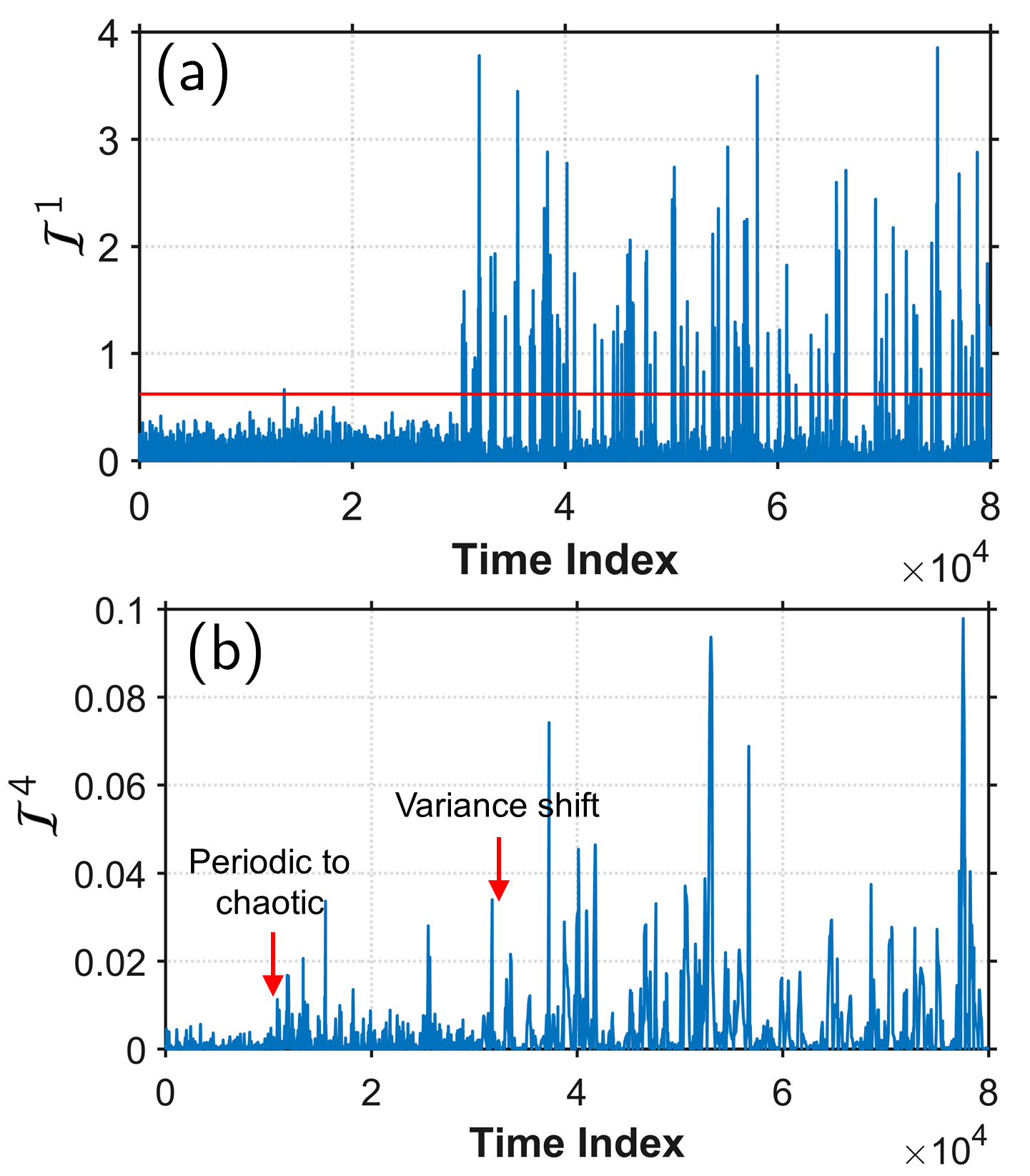}
	\centering
	\caption{InSync statistic as obtained by using the (a) first cluster with rotation components, $R^1(t), R^2(t)~ \& ~R^3(t)$ and (b) second cluster with rotation components $R^4(t) ~ \& ~R^5(t)$.}
	\label{figure:4}	
\end{figure}

 \begin{figure*}[!t] 
	\includegraphics[width = 0.75\textwidth]{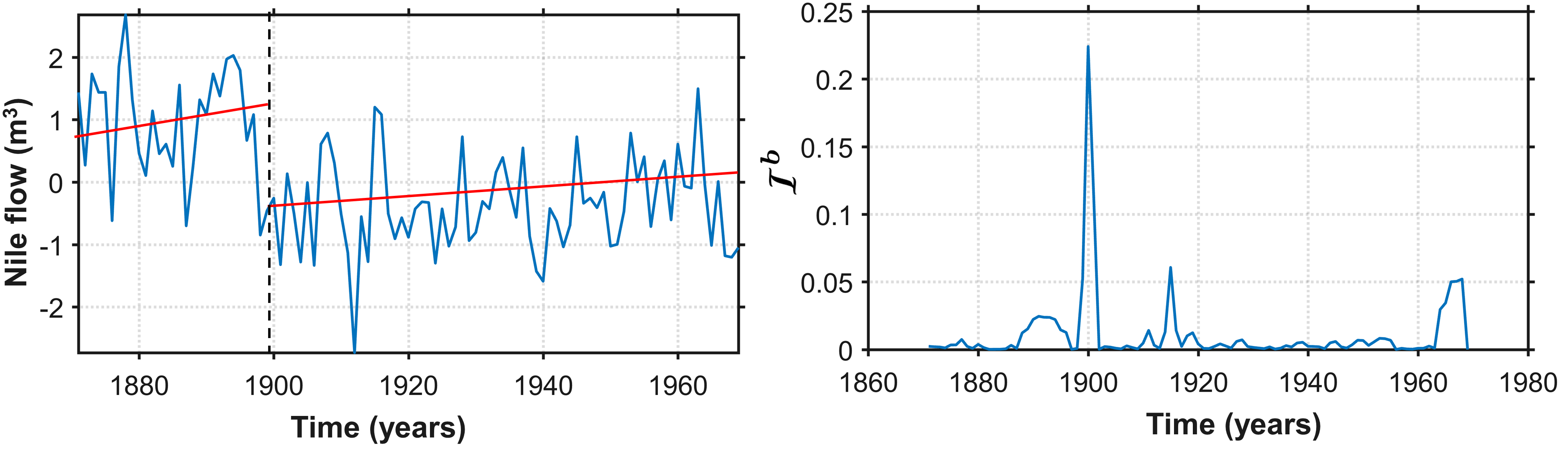}
	\centering
	\caption{(a)  Time series showing the normalized Nile flow rate; (b) shows the InSync statistic $\mathcal{I}(\hslash^4_{k}(t))$ with the set $\mathcal{G}$ being $\{R^6(t), R^7(t), R^8(t)\}$.}
	\label{figure:7}
\end{figure*}
First, we utilize the mutual agreement concept to identify the cluster of rotation components that may preserve the change point information. As shown in Fig.~\ref{figure:3}, we notice that there are three different clusters of rotation components with mutual agreement values $\mathcal{M}^1 =0.52$ consisting of rotation components, $R^1(t), R^2(t)~ \& ~R^3(t)$; $\mathcal{M}^2 = 0.21 $ with components $R^4(t)~ \& ~R^5(t)$ and $\mathcal{M}^3 = 0.13 $ with components $R^6(t), R^7(t)~ \& ~R^8(t)$. We begin with the first cluster of rotation components $R^1(t), R^2(t)~ \& ~R^3(t)$. Intuitively, the cluster with levels of rotation component $j\leq 3$ should capture the singularities (i.e., the high frequency change point features). Therefore, using the ARL$0\approx 370$ or equivalently the specificity for the in-control region as 0.9973, we set the threshold on the InSync statistic (shown in red in Fig.~S4(a)) and estimate the sensitivity of detecting the singularity points. Based on this threshold, we get a sensitivity of 0.9925.

Next, we identify the persistent change points (first periodic to chaotic and then the variance shift) sequentially by reseting the threshold and the in-control region, once a change point is detected. To identify the first change point, we again set ARL$0\approx 370$ based on the first 5000 data points and estimate the ARL1 value from the resulting CUSUM chart. The ARL1 value for this case is 1.11. We reset the monitoring statistic and the second change point at $t = 30000$ is detected with an ARL1 = 1.10. The third cluster consisting of rotation components, $R^6(t), R^7(t)~ \& ~R^8(t)$ also resulted in the same conclusion as that of the second cluster consisting of rotation components $R^4(t)$ and $R^5(t)$ but with different ARL1 values.

\noindent \textbf{Trend change detection: }To test the performance of our method for detecting trend changes, we consider the annual Nile river flow measured at Aswan from 1871 to 1970. Several historical records and research have suggested that the trend in the mean flow level shifted after 1897. The shift in the trend in shown in Fig.~\ref{figure:7}(a). In this case study we implement the InSync statistic to determine the changepoint. To implement the proposed methodology, we first determined the base component $R^b(t)$ from the network representation as shown in Section 3.2 of the main document. Here, the cluster of rotation components with maximum mutual agreement is $\mathcal{G} = \{R^6(t), R^7(t), R^8(t)\}$ with $R^6(t)$ as the base component.  The InSync statistic $\mathcal{I}(\hslash^6_{k}(t))$ for every halfwave defined about the base component $R^6(t)$ is shown in Fig.~\ref{figure:7}(b). The shift in the trend is captured by the InSync statistic as a sharp peak around 1900. To compare the performance of the proposed method, we compared the ARL1 values from EWMA, WCUSUM, Pruned Exact Linear Time (PELT) from the \textit{CPM} package, and the likelihood ratio test (LRT) from \textit{changepoint} package. We notice that in all the cases, InSync statistic was able to consistently detect the change point with ARL1 value as reported in Table~\ref{table:t1}. 

\begin{figure}[!t]
\includegraphics[width = 0.35\textwidth]{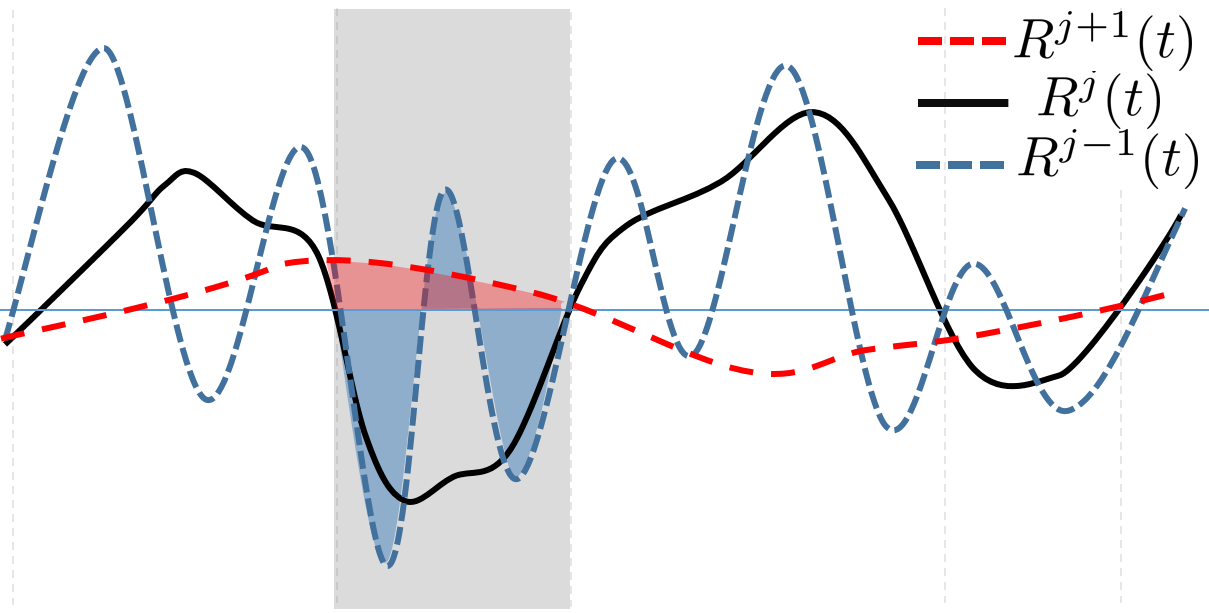}
\centering
\caption{ Illustrative example to show the halfwave span property of the rotation components, $R^j(t)$. We notice that the support of $\hslash^j_k(t)$, i.e., $(z_k^{j},z_{k+1}^{j}]$ in $R^j(t)$ spans 3 halfwaves from the previous level, $R^{j-1}(t)$ and a fraction of halfwave from the next level rotation component, $R^{j+1}(t)$.}
\label{figure:5}	
\end{figure}

\begin{table}[t]
	\renewcommand{\arraystretch}{1.25}
	\caption{Comparison of ARL1 for different methods considered in this case study}
	\label{table_example}
	\centering	
	\begin{threeparttable}
		\begin{tabularx}{0.45\textwidth}{c *{6}{Y}}
			\hline
			EWMA & WCUSUM & PELT & LRT  & InSync\\ 
			
			2.98 & 300 & 1.0056 & 1  & 1.0007\\
			\hline
		\end{tabularx}
	\end{threeparttable}
	\label{table:t1}
\end{table}
\section{ Proof of proposition 3}
\label{appendix:H}
\setcounter{proposition}{3}
\begin{proposition} [\textbf{Property 3}] The support of $\hslash^j_{k}(t)$ at any level $j\geq2$ spans at least one halfwave $\hslash^i_{k}(t)$ from its sub-level $\{R^i(t)\}_{i<j}$ and at most one $\hslash^f_{k}(t) $ from its super-level $\{R^f(t)\}_{f>j}$ as shown in Fig.~\ref{figure:5}.
\end{proposition} 

\begin{proof}{}\label{one}
\normalfont {Here, the support, $\text{supp}\left(\hslash^j_{k}(t)\right)=[z_k,z_{k+1}]$. Based on Property S1, it follows that the extrema $\{\tau_k\}_{k=1,2,\ldots,N}$ at level $j$ evolves in level $j+1$ according to an  extrema vanishing transition (saddle-node and pitchfork) or an extrema preserving transition (trans-critical and pitchfork) \cmt{respectively} (cf. Fig.~\ref{figure:5}). No new extrema are created. Therefore, $\text{supp}\left(\hslash^j_{k}(t)\right) \leq \text{supp}\left(\hslash^{j+1}_{k}(t)\right)$, and thus, $\hslash^{j}_{k}(t)$ spans at most one $\hslash^{j+1}_{k}(t) $ from its super-level $\{R^f(t)\}_{f>j}$. Similarly, we have $ \text{supp}\left(\hslash^j_{k}(t)\right) \geq \text{supp}\left(\hslash^{j-1}_{k}\right)$ and hence $\hslash^j_{k}$ spans at least one $\hslash^{j-1}_{k}$ from its sub-level $\{R^i(t)\}_{i<j}$.}
\end{proof}

\section{Derivation of the Distribution function of InSync}
\label{appendix:I}
\begin{proposition} The distribution function of the InSync statistic, $\mathcal{I}(\hslash^b_{k}(t))\equiv {\sum_j g(\mathcal{E}_j)}\phi$, considering two arbitrary levels $j = 1,2$, can be expressed as the following product distribution:
	\begin{equation*}
	F_{\mathcal{I}}(\iota) \propto\int_{-\infty}^{\infty}\frac{1}{2}f_{\sum g(\mathcal{E})}\left(\frac{\mathcal{\iota}}{\phi}\right)\frac{1}{|\phi|}d\phi
	\end{equation*} where the energy term, $\sum_{j=1,2} g(\mathcal{E}_j)$ follows a generalized Pareto (GP) distribution with scale, shape and location parameters given as, $(c_1s_1+c_2s_2)/(s_1+s_2)$, $(1/s_1+1/s_2)$ and 0, respectively and the phase term, $\phi\sim U(-1,1)$. $\{c_1, c_2\}, \{s_1, s_2\}$ denote the scale and shape parameters of GP distributions representing $g(\mathcal{E}_j),j=1,2$.
\end{proposition} 
\begin{figure}[!b]
	\includegraphics[width = 0.35\textwidth]{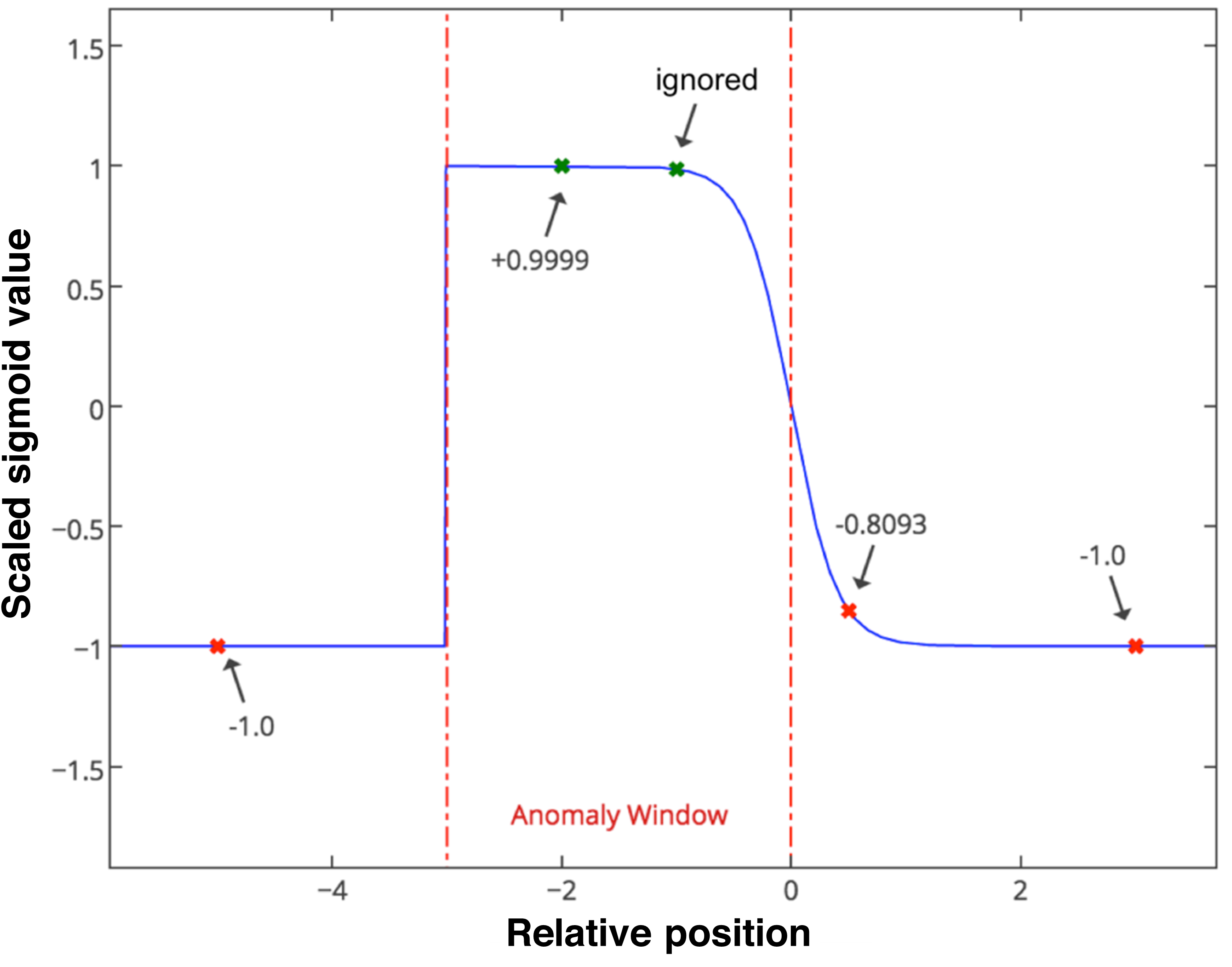}
	\centering
	\caption{Scale sigmoid function to determine the score relative to the anomaly window shown by the red dotted lines. The markers represent the detections made by a sample algorithm relative to the anomaly window.}
	\label{figure:6}	
\end{figure}
\begin{proof}
	\normalfont {First, we determine the distribution function of ${\sum_j g(\mathcal{E}_j)}$ where $g(.)$ is the exponential transform of the energy $\mathcal{E}_j$ of halfwave in level $j$. To derive the distribution function, we refer to the signal $x_k$ given as: $$x_k = (-1)^k|w_k|; w_k \sim \mathcal{N}(0,\sigma^2),k\in \mathbb{Z}^+$$ where the magnitudes of successive samples (alternating extrema) are drawn from a white noise process. To define a halfwave, we consider three consecutive extrema points. As the extrema points are normally distributed, each of the halfwave is essentially a normal random vector with 3 components. Let $\hslash^j$ represent the halfwaves such that $\hslash^j\sim \mathcal{N}(0,I_3\sigma^2_j)$ at levels $j=1,2$. Therefore, $$\mathcal{E}_j\equiv \sum_{k=1}^{3}(\hslash^j_k)^2\sim \chi^2_3, ~j = 1,2$$From [3], we note that the exponential transform of a random variable from exponential family follows a generalized Pareto (GP) distribution. Therefore, $g(\mathcal{E}_j)$ follows a GP distribution with threshold parameter $t_j =1$ and scale and shape given by $s_j$ and $c_j$, respectively. The parameters $c_j$ and $s_j$ can be estimated using the method of moments and is given as: 	
		\begin{align*}
		\frac{c_j}{1+s_j} & =\mathbb{E}(g(\mathcal{E}_j)) \\ &=\mathbb{E}\left(1+\alpha\mathcal{E}_j+\frac{(\alpha\mathcal{E}_j)^2}{2!}+\frac{(\alpha\mathcal{E}_j)^3}{3!}+\ldots\right)
		\end{align*} 
	\begin{align*}
		\frac{c^2_j}{(1+s_j)^2(1+2s_j)}&=\mathbb{E}(g(\mathcal{E}_j)^2)-\mathbb{E}(g(\mathcal{E}_j))^2 \\ &= \mathbb{E}\left(1+2\alpha\mathcal{E}_j+\ldots\right) - \left(\frac{c_j}{1+s_j}\right)^2
	\end{align*}	For the distribution function of the phase synchronization component, $\phi$, we note that this is nothing but the normalized inner product of two normal vectors and follows $U(-1,1)$. Finally, we note that the InSync statistic is the product of $g(\mathcal{E}_j)$ and $\phi$, and therefore, follows a product distribution given as, $$	~~~~~~~~~~~~~~~~F_{\mathcal{I}}(\iota) \propto\int_{-\infty}^{\infty}\frac{1}{2}f_{\sum g(\mathcal{E})}\left(\frac{\mathcal{\iota}}{\phi}\right)\frac{1}{|\phi|}d\phi~~~~~~~~~~~~~~\QED$$ 	
	}
\end{proof}

%

\section{Scoring function employed in Section 4.5 }
\label{appendix:J}

The performance of various anomaly detection algorithms, as proposed in [4], is assessed based on a standard scoring function. The function assigns a weighted positive score to an algorithm that is able to detect a change within a prescribed anomaly window and penalizes for any missing anomaly. Length of the anomaly window is set to 10\% of the length of time series divided by the total number of anomalies in the dataset. 

Correctly identified anomalies are assigned a score, $A_{\text{TP}}$ of 1 while false positives and missed anomalies are penalized with scores of $A_{{\text{FP}}}=0.11$ and $A_{{\text{FN}}}=1$, respectively. Depending on where the anomaly was detected with reference to the anomaly window, the reward as well as the penalty weights assigned to individual detection are determined using a scaled sigmoid function as shown in Fig.~\ref{figure:6}. From the figure we notice that a negative weight of $-$1 is assigned to a FP which is far from the anomaly window. In contrast, a FP is penalized less if it is closer to the window. 

An illustrative example for a sample anomaly window is shown in Fig.~\ref{figure:6}. The first
point is a FP preceding the anomaly window and is penalized with a weight of $-$1. Next, for the two detections within the anomaly window, we only count the earliest TP and is assigned a positive weight of 1. Following the anomaly window, we notice two FPs. Since the first FP is less detrimental because it is close to the window as compared to the second, hence a relatively smaller negative weight of $-$0.83 based on the sigmoid function is assigned to the former. In contrast, a weight of $-$1 is assigned to the latter because it’s too far from the window to be associated with the true anomaly. True negatives are assumed to make no contributions. Hence, with the scores for FP, TP and FN as mentioned earlier, the final score for the example shown in Fig.~\ref{figure:6} is: $-1.0A_{{\text{FP}}} + 1A_{\text{TP}} -0.8093A_{{\text{FP}}} -1.0A_{{\text{FP}}} = 0.6909$.

\section{REFERENCES}

\begin{enumerate}[label={[\arabic*]}]
\item J. M. Restrepo, S. Venkataramani, D. Comeau, and H. Flaschka, ``Defining a trend for time series using the intrinsic time-scale decomposition,'' New
J. Phys., vol. 16, no. 8, p. 085004, Aug. 2014.
\item L. A. Goodman, ``On the exact variance of products,'' Journal of the American statistical association, vol. 55, no. 292, pp. 708–713, Dec. 1960.
\item W. Hürlimann, ``General affine transform families: why is the pareto an exponential transform?'' Statistical Papers, vol. 44, no. 4, pp. 499–518,
Oct. 2003.
\item S. Ahmad, A. Lavin, S. Purdy, and Z. Agha, ``Unsupervised real-time anomaly detection for streaming data,'' Neurocomputing, vol. 262, pp. 134–147, Nov. 2017.
\end{enumerate}

%

\end{document}